\documentclass[12pt,3p]{nj_style}

\usepackage{amsmath}
\usepackage{amssymb}
\usepackage{amsthm}
\usepackage{booktabs}
\usepackage{enumerate}
\usepackage[font=small,labelfont=bf,margin=8pt]{caption}
\usepackage[usenames,dvipsnames]{color}
\usepackage{float}\restylefloat{table} 
\usepackage{graphicx}
\usepackage[colorlinks=true]{hyperref}
\usepackage{pifont}
\usepackage[table]{xcolor}
\newcommand{\cmark}{\checkmark}

\usepackage{tikz}
\tikzstyle{vertex}=[circle, draw, inner sep=0pt, minimum size=4pt]
\newcommand{\vertex}{\node[vertex]}
\DeclareTextSymbolDefault{\DH}{T1} 
\definecolor{redback}{RGB}{255,168,168}
\definecolor{greenback}{RGB}{168,255,168}
\definecolor{yellowback}{RGB}{255,255,168}

\newtheorem{thm}{Theorem} 

\newtheorem{lemma}[thm]{Lemma}
\newtheorem{prop}[thm]{Proposition}

\newcommand{\ket}[1]{| #1 \rangle}
\newcommand{\braket}[2]{\langle #1 | #2 \rangle}
\newcommand{\ketbra}[2]{| #1 \rangle\langle #2 |}

\newcommand{\bb}[1]{\mathbb{#1}}
\newcommand{\cl}[1]{\mathcal{#1}}

\makeatletter
\renewcommand*\env@matrix[1][c]{\hskip -\arraycolsep
  \let\@ifnextchar\new@ifnextchar
  \array{*\c@MaxMatrixCols #1}}
\makeatother

\begin{document}

\begin{frontmatter}


\title{The Structure of Qubit Unextendible Product Bases}

\author[IQC]{Nathaniel Johnston}
\ead{nathaniel.johnston@uwaterloo.ca}

\address[IQC]{Institute for Quantum Computing, University of Waterloo, Waterloo, Ontario N2L~3G1, Canada}

\begin{abstract}
	Unextendible product bases have been shown to have many important uses in quantum information theory, particularly in the qubit case. However, very little is known about their mathematical structure beyond three qubits. We present several new results about qubit unextendible product bases, including a complete characterization of all four-qubit unextendible product bases, which we show there are exactly 1446 of. We also show that there exist $p$-qubit UPBs of almost all sizes less than $2^p$.
\end{abstract}


\end{frontmatter}

\section{Introduction}

Unextendible product bases (UPBs) are one of the most useful and versatile objects in the theory of quantum entanglement. While they were originally introduced as a tool for constructing bound entangled states \cite{BDFMRSSW99,DMSST03}, they can also be used to construct indecomposible positive maps \cite{Ter01} and to demonstrate the existence of nonlocality without entanglement---that is, they can not be perfectly distinguished by local quantum operations and classical communication, even though they contain no entanglement. Furthermore, in the qubit case (i.e., the case where each local space has dimension $2$), unextendible product bases can be used to construct tight Bell inequalities with no quantum violation \cite{AFKKPLA12,ASHKLA11} and subspaces of small dimension that are locally indistinguishable \cite{DXY10}.

Despite their many uses, very little is known about the mathematical structure of unextendible product bases. For example, UPBs have only been completely characterized in $\mathbb{C}^2 \otimes \mathbb{C}^n$ (where all UPBs are trivial in the sense that they span the entire space \cite{BDMSST99}), $\mathbb{C}^3 \otimes \mathbb{C}^3$ (where all UPBs belong to a known six-parameter family \cite{DMSST03}), and $\mathbb{C}^2 \otimes \mathbb{C}^2 \otimes \mathbb{C}^2$ (where there is only one nontrivial UPB up to local operations \cite{Bra04}). The goal of the present paper is to thoroughly investigate the structure of qubit unextendible product bases (i.e., UPBs in $(\mathbb{C}^2)^{\otimes p}$ for some $p \in \mathbb{N}$).

Our first contribution is to completely characterize all unextendible product bases on four qubits. Unlike the three qubit case, where all nontrivial UPBs are essentially the same (in the sense that they all have the same orthogonality graph), we show that nontrivial UPBs on four qubits can have one of exactly 1446 different orthogonality graphs, and hence the set of qubit UPBs quickly becomes very complicated as the number of qubits increases.

We also consider UPBs on larger numbers of qubits. In particular, we address the question of how many states a $p$-qubit UPB can have. The minimum number of states in such a UPB is known to always be between $p+1$ and $p+4$ inclusive \cite{Joh13UPB}, and the results of \cite{CD13} immediately imply that the maximum number of states is $2^p-4$ (or $2^p$ if we allow trivial UPBs that span the entire $2^p$-dimensional space). However, very little has been known about what intermediate sizes can be attained as the cardinality of some $p$-qubit UPB.

Surprisingly, we show that there are intermediate sizes that are \emph{not} attainable as the cardinality of any $p$-qubit UPB (contrast this with the case of non-orthogonal UPBs, which exist of any size from $p+1$ to $2^p$ inclusive \cite{Bha06}). However, we show that these cases are rare in the sense that, as $p \rightarrow \infty$, the proportion of intermediate sizes that are attainable by some $p$-qubit UPB goes to $1$. Furthermore, we show that all unattainable sizes are very close to either the minimal or maximal size, and we provide examples to demonstrate that both of these cases are possible.

The paper is organized as follows. In Section~\ref{section:prelims} we introduce some basic facts about UPBs that will be of use for us, and present the mathematical tools that we will use to prove our results. We then describe our characterization of four qubit UPBs in Section~\ref{section:4qubit}, which was found via computer search (described in Appendix~A). We also discuss some new UPBs on five and six qubits that were found via the same computer search in Section~\ref{section:5qubit}. Finally, we consider the many-qubit case in Section~\ref{section:manyqubit}, where we show that there exist qubit UPBs of most (but not all) sizes between the minimal and maximal size.

\section{Preliminaries}\label{section:prelims}

A $p$-qubit pure quantum state is represented by a unit vector $\ket{v} \in (\bb{C}^{2})^{\otimes p}$, which is called a \emph{product state} if it can be decomposed in the following form:
\begin{align*}
	\ket{v} = \ket{v_1} \otimes \cdots \otimes \ket{v_p} \ \ \text{ with } \ \ \ket{v_j} \in \bb{C}^{2} \ \forall \, j.
\end{align*}
The standard basis of $\mathbb{C}^2$ is $\{\ket{0},\ket{1}\}$ and we use $\{\ket{a},\ket{\overline{a}}\}$, $\{\ket{b},\ket{\overline{b}}\}$, $\{\ket{c},\ket{\overline{c}}\}, \ldots$ to denote orthonormal bases of $\mathbb{C}^2$ that are different from $\{\ket{0},\ket{1}\}$ and from each other (i.e., $\ket{a} \neq \ket{0},\ket{1},\ket{b},\ket{\overline{b}},\ket{c},\ket{\overline{c}}$, and so on). We also will sometimes find it useful to omit the tensor product symbol when discussing multi-qubit states. For example, we use $\ket{0a11\overline{a}}$ as a shorthand way to write $\ket{0}\otimes\ket{a}\otimes\ket{1}\otimes\ket{1}\otimes\ket{\overline{a}}$.

A $p$-qubit \emph{unextendible product basis (UPB)} \cite{BDMSST99,DMSST03} is a set $\mathcal{S} \subseteq (\bb{C}^{2})^{\otimes p}$ satisfying the following three properties:
\begin{enumerate}[(a)]
	\item every $\ket{v} \in \mathcal{S}$ is a product state;
	
	\item $\braket{v}{w} = 0$ for all $\ket{v} \neq \ket{w} \in \mathcal{S}$; and
	
	\item for all product states $\ket{z} \notin \mathcal{S}$, there exists $\ket{v} \in \mathcal{S}$ such that $\braket{v}{z} \neq 0$.
\end{enumerate}
That is, a UPB is a set of mutually orthogonal product states such that there is no product state orthogonal to every member of the set. To be explicit, when we refer to the ``size'' of a UPB, we mean the number of states in the set.

We now present a result that shows how to use known UPBs to construct larger UPBs on more qubits. This result is well-known and follows easily from \cite[Lemma~2.3]{Fen06}, but we make repeated use of it and thus prove it explicitly.
\begin{prop}\label{prop:qubit_add_together}
	Let $p \in \mathbb{N}$. If there exist $p$-qubit UPBs $\mathcal{S}_1$ and $\mathcal{S}_2$ with $|\mathcal{S}_1| = s_1$ and $|\mathcal{S}_2| = s_2$ then there exists a $(p+1)$-qubit UPB $\mathcal{S}$ with $|\mathcal{S}| = s_1 + s_2$.
\end{prop}
\begin{proof}
	If we write $\mathcal{S}_1 = \{\ket{v_1},\ldots,\ket{v_{s_1}}\}$ and $\mathcal{S}_2 = \{\ket{w_1},\ldots,\ket{w_{s_2}}\}$ then it is straightforward to see that the set
	\begin{align*}
		\mathcal{S} := \big\{ \ket{v_1}\otimes\ket{0}, \ldots, \ket{v_{s_1}}\otimes\ket{0}, \ket{w_1}\otimes\ket{1}, \ldots, \ket{w_{s_2}}\otimes\ket{1} \big\} \subset (\mathbb{C}^2)^{\otimes(p+1)}
	\end{align*}
	satisfies properties~(a) and~(b) of a UPB. We prove that it also satisfies property~(c) by contradiction: suppose that there were a product state $\ket{z} \in (\mathbb{C}^2)^{\otimes(p+1)}$ such that $\braket{v}{z} = 0$ for all $\ket{v} \in \mathcal{S}$. If we write $\ket{z} = \ket{z_{1\ldots p}} \otimes \ket{z_{p+1}}$ for some product state $\ket{z_{1\ldots p}} \in (\mathbb{C}^2)^{\otimes p}$ and $\ket{z_{p+1}} \in \mathbb{C}^2$ then we have $\braket{v_j}{z_{1\ldots p}}\braket{0}{z_{p+1}} = 0$ for all $1 \leq j \leq s_1$. Unextendibility of $\mathcal{S}_1$ implies that $\braket{0}{z_{p+1}} = 0$. However, a similar argument using unextendibility of $\mathcal{S}_2$ shows that $\braket{1}{z_{p+1}} = 0$, which implies that $\ket{z_{p+1}} = 0$, which is the contradiction that completes the proof.
\end{proof}

\subsection{Orthogonality Graphs}\label{section:orthog_graphs}

One tool that we find helps to visualize UPBs and simplify proofs is an orthogonality graph. Given a set of product states $\cl{S} = \{\ket{v_1},\ldots,\ket{v_{s}}\} \subseteq (\bb{C}^2)^{\otimes p}$ with $|\cl{S}| = s$, the \emph{orthogonality graph of $\cl{S}$} is the graph on $s$ vertices $V := \{v_1,\ldots,v_{s}\}$ such that there is an edge $(v_i,v_j)$ of color $\ell$ if and only if $\ket{v_i}$ and $\ket{v_j}$ are orthogonal to each other on qubit $\ell$. Rather than actually using $p$ colors to color the edges of the orthogonality graph, for ease of visualization we instead draw $p$ different graphs on the same set of vertices---one for each qubit (see Figure~\ref{fig:2dim_eg}).

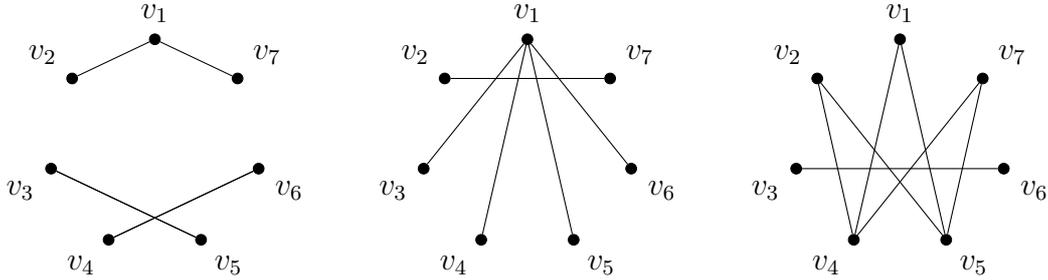
\begin{figure}[htb]
	\centering
	\begin{tikzpicture}[x=1.4cm, y=1.4cm, label distance=0cm]     
		\vertex[fill] (v00) at (0,1) [label=90:$v_{1}$]{};
		\vertex[fill] (v01) at (-0.777,0.629) [label=141:$v_{2}$]{};
		\vertex[fill] (v02) at (-0.975,-0.225) [label=193:$v_{3}$]{};
		\vertex[fill] (v03) at (-0.434,-0.899) [label=244:$v_{4}$]{};
		\vertex[fill] (v04) at (0.434,-0.899) [label=296:$v_{5}$]{};
		\vertex[fill] (v05) at (0.975,-0.225) [label=347:$v_{6}$]{};
		\vertex[fill] (v06) at (0.777,0.629) [label=39:$v_{7}$]{};
				
		\vertex[fill] (v10) at (3.5,1) [label=90:$v_{1}$]{};
		\vertex[fill] (v11) at (2.723,0.629) [label=141:$v_{2}$]{};
		\vertex[fill] (v12) at (2.525,-0.225) [label=193:$v_{3}$]{};
		\vertex[fill] (v13) at (3.066,-0.899) [label=244:$v_{4}$]{};
		\vertex[fill] (v14) at (3.934,-0.899) [label=296:$v_{5}$]{};
		\vertex[fill] (v15) at (4.475,-0.225) [label=347:$v_{6}$]{};
		\vertex[fill] (v16) at (4.277,0.629) [label=39:$v_{7}$]{};

		\vertex[fill] (v20) at (7,1) [label=90:$v_{1}$]{};
		\vertex[fill] (v21) at (6.223,0.629) [label=141:$v_{2}$]{};
		\vertex[fill] (v22) at (6.025,-0.225) [label=193:$v_{3}$]{};
		\vertex[fill] (v23) at (6.566,-0.899) [label=244:$v_{4}$]{};
		\vertex[fill] (v24) at (7.434,-0.899) [label=296:$v_{5}$]{};
		\vertex[fill] (v25) at (7.975,-0.225) [label=347:$v_{6}$]{};
		\vertex[fill] (v26) at (7.777,0.629) [label=39:$v_{7}$]{};

		\path 
			(v00) edge (v01)
			(v00) edge (v06)
			(v02) edge (v04)
			(v02) edge (v04)
			(v03) edge (v05)
			(v03) edge (v05)
			
			(v11) edge (v16)
			(v10) edge (v12)
			(v10) edge (v13)
			(v10) edge (v14)
			(v10) edge (v15)

			(v22) edge (v25)
			(v20) edge (v23)
			(v20) edge (v24)
			(v21) edge (v23)
			(v21) edge (v24)
			(v26) edge (v23)
			(v26) edge (v24)
		;
	\end{tikzpicture}
	\caption{The orthogonality graph of a set of $7$ product states in $(\bb{C}^2)^{\otimes 3}$. The left graph indicates that state $\ket{v_1}$ is orthogonal to $\ket{v_2}$ and $\ket{v_7}$ on the first qubit, $\ket{v_3}$ is orthogonal to $\ket{v_5}$ on the first qubit, and $\ket{v_4}$ is orthogonal to $\ket{v_6}$ on the first qubit (and the middle and right graphs similarly describe the orthogonalities on the second and third qubits). These states do not form a UPB, since (for example) there is no edge between $v_2$ and $v_3$ in any of the graphs, so the states $\ket{v_2}$ and $\ket{v_3}$ are not orthogonal.}\label{fig:2dim_eg}
\end{figure}

The requirement~(b) that the members of a UPB are mutually orthogonal is equivalent to requiring that every edge is present on at least one qubit in its orthogonality graph (in other words, the orthogonality graph is an edge coloring of the complete graph). The unextendibility condition~(c) is more difficult to check, so we first need to make some additional observations. In particular, it is important to notice that if $\ket{z_1},\ket{z_2},\ket{z_3} \in \bb{C}^2$ satisfy $\braket{z_1}{z_2} = \braket{z_1}{z_3} = 0$, then it is necessarily the case that $\ket{z_2} = \ket{z_3}$ (up to an irrelevant scalar multiple). There are two important consequences of this observation:
\begin{enumerate}
	\item The orthogonality graph associated with any individual qubit in a product basis of $(\mathbb{C}^2)^{\otimes p}$ is the disjoint union of complete bipartite graphs. For example, the orthogonality graph of the first qubit in Figure~\ref{fig:2dim_eg} is the disjoint union of $K_{2,1}$ and two copies of $K_{1,1}$, the orthogonality graph of the second qubit is the disjoint union of $K_{4,1}$ and $K_{1,1}$, and the orthogonality graph of the third qubit is the disjoint union of $K_{3,2}$ and $K_{1,1}$.

	\item We can determine whether or not a set of qubit product states forms a UPB entirely from its orthogonality graph (a fact that is not true when the local dimensions are larger than $2$ \cite{DMSST03}). For this reason, we consider two qubit UPBs to be \emph{equivalent} if they have the same orthogonality graphs up to permuting the qubits and relabeling the vertices (alternatively, we consider two qubit UPBs to be equivalent if we can permute qubits and change each basis of $\mathbb{C}^2$ used in the construction of one of the UPBs to get the other UPB).
\end{enumerate}

Following \cite{Joh13UPB}, we sometimes draw orthogonality graphs in a form that makes their decomposition in terms of complete bipartite graphs more transparent---we draw shaded regions indicating which states are equal to each other (up to scalar multiple) on the given qubit, and lines between shaded regions indicate that all states in one of the regions are orthogonal to all states in the other region on that qubit (see Figure~\ref{fig:2dim_eg_compact}).
\begin{figure}[htb]
	\centering
	\begin{tikzpicture}[x=1.4cm, y=1.4cm, label distance=0.07cm]
    \draw[draw=black] (0,0.5) -- (0,1);

    \draw[line width=0.02cm,draw=black,fill=lightgray] (0,1) circle (0.2cm);
    \filldraw[fill=lightgray,line width=0.44cm,line join=round,draw=black] (-0.777,0.629) -- (-0.1,0.5) -- (0.1,0.5) -- (0.777,0.629) -- (0.1,0.5) -- (-0.1,0.5) -- cycle;
    \filldraw[fill=lightgray,line width=0.4cm,line join=round,draw=lightgray] (-0.777,0.629) -- (-0.1,0.5) -- (0.1,0.5) -- (0.777,0.629) -- (0.1,0.5) -- (-0.1,0.5) -- cycle;

    \draw[draw=black] (-0.975,-0.225) -- (0.434,-0.899);
    \draw[draw=black] (-0.434,-0.899) -- (0.975,-0.225);
    
    \draw[line width=0.02cm,draw=black,fill=lightgray] (-0.975,-0.225) circle (0.2cm);
    \draw[line width=0.02cm,draw=black,fill=lightgray] (-0.434,-0.899) circle (0.2cm);
    \draw[line width=0.02cm,draw=black,fill=lightgray] (0.434,-0.899) circle (0.2cm);
    \draw[line width=0.02cm,draw=black,fill=lightgray] (0.975,-0.225) circle (0.2cm);
    
    \draw[draw=black] (2.723,0.629) -- (4.277,0.629);
    \draw[draw=black] (3.5,1) -- (3.5,-0.7);

    \filldraw[fill=lightgray,line width=0.44cm,line join=round,draw=black] (2.525,-0.225) -- (3.066,-0.899) -- (3.934,-0.899) -- (4.475,-0.225) -- cycle;
    \filldraw[fill=lightgray,line width=0.4cm,line join=round,draw=lightgray] (2.525,-0.225) -- (3.066,-0.899) -- (3.934,-0.899) -- (4.475,-0.225) -- cycle;

    \draw[line width=0.02cm,draw=black,fill=lightgray] (3.5,1) circle (0.2cm);
    \draw[line width=0.02cm,draw=black,fill=lightgray] (2.723,0.629) circle (0.2cm);
    \draw[line width=0.02cm,draw=black,fill=lightgray] (4.277,0.629) circle (0.2cm);

    \draw[draw=black] (6.025,-0.225) -- (7.975,-0.225);
    \draw[draw=black] (7,1) -- (7,-0.899);

    \draw[line width=0.02cm,draw=black,fill=lightgray] (6.025,-0.225) circle (0.2cm);
    \draw[line width=0.02cm,draw=black,fill=lightgray] (7.975,-0.225) circle (0.2cm);

    \filldraw[fill=lightgray,line width=0.44cm,line join=round,draw=black] (7,1) -- (6.223,0.629) -- (7.777,0.629) -- cycle;
    \filldraw[fill=lightgray,line width=0.4cm,line join=round,draw=lightgray] (7,1) -- (6.223,0.629) -- (7.777,0.629) -- cycle;

    \filldraw[fill=lightgray,line width=0.44cm,line join=round,draw=black] (6.566,-0.899) -- (7.434,-0.899) -- (6.566,-0.899) -- cycle;
    \filldraw[fill=lightgray,line width=0.4cm,line join=round,draw=lightgray] (6.566,-0.899) -- (7.434,-0.899) -- (6.566,-0.899) -- cycle;
    
    \vertex[fill] (v00) at (0,1) [label=90:$v_{1}$]{};
		\vertex[fill] (v01) at (-0.777,0.629) [label=141:$v_{2}$]{};
		\vertex[fill] (v02) at (-0.975,-0.225) [label=193:$v_{3}$]{};
		\vertex[fill] (v03) at (-0.434,-0.899) [label=244:$v_{4}$]{};
		\vertex[fill] (v04) at (0.434,-0.899) [label=296:$v_{5}$]{};
		\vertex[fill] (v05) at (0.975,-0.225) [label=347:$v_{6}$]{};
		\vertex[fill] (v06) at (0.777,0.629) [label=39:$v_{7}$]{};
				
		\vertex[fill] (v10) at (3.5,1) [label=90:$v_{1}$]{};
		\vertex[fill] (v11) at (2.723,0.629) [label=141:$v_{2}$]{};
		\vertex[fill] (v12) at (2.525,-0.225) [label=193:$v_{3}$]{};
		\vertex[fill] (v13) at (3.066,-0.899) [label=244:$v_{4}$]{};
		\vertex[fill] (v14) at (3.934,-0.899) [label=296:$v_{5}$]{};
		\vertex[fill] (v15) at (4.475,-0.225) [label=347:$v_{6}$]{};
		\vertex[fill] (v16) at (4.277,0.629) [label=39:$v_{7}$]{};

		\vertex[fill] (v20) at (7,1) [label=90:$v_{1}$]{};
		\vertex[fill] (v21) at (6.223,0.629) [label=141:$v_{2}$]{};
		\vertex[fill] (v22) at (6.025,-0.225) [label=193:$v_{3}$]{};
		\vertex[fill] (v23) at (6.566,-0.899) [label=244:$v_{4}$]{};
		\vertex[fill] (v24) at (7.434,-0.899) [label=296:$v_{5}$]{};
		\vertex[fill] (v25) at (7.975,-0.225) [label=347:$v_{6}$]{};
		\vertex[fill] (v26) at (7.777,0.629) [label=39:$v_{7}$]{};
	\end{tikzpicture}
	\caption{A representation of the same orthogonality graph as that of Figure~\ref{fig:2dim_eg}. Vertices within the same shaded region represent states that are equal to each other on that qubit. Lines between shaded regions indicate that every state within one of the regions is orthogonal to every state within the other region on that qubit.}\label{fig:2dim_eg_compact}
\end{figure}
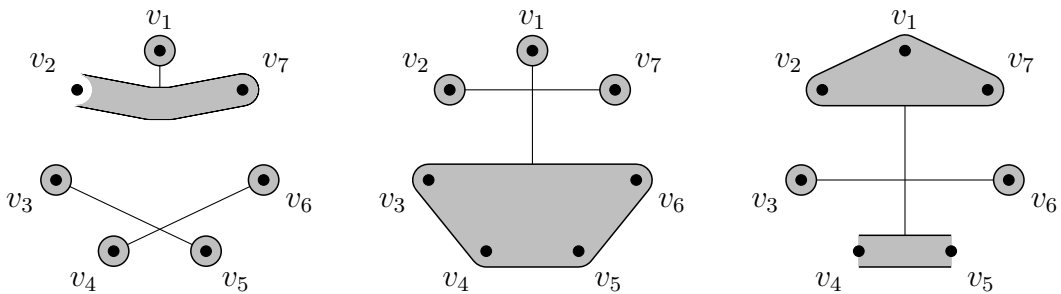

\subsection{UPBs on Three or Fewer Qubits}\label{section:fewqubits}

We now review what is known about UPBs in the space $(\mathbb{C}^2)^{\otimes p}$ when $1 \leq p \leq 3$. It is well-known that there are no nontrivial qubit UPBs when $p \leq 2$ \cite{BDMSST99}, so the first case of interest is when $p = 3$. In this case, the \textbf{Shifts} UPB \cite{BDMSST99} provides one of the oldest examples of a nontrivial UPB and consists of the following four states:
\begin{align*}
	\mathbf{Shifts} := \big\{ \ket{000}, \ket{1{+}{-}}, \ket{{-}1{+}}, \ket{{+}{-}1} \big\},
\end{align*}
where $\ket{+} := (\ket{0}+\ket{1})/\sqrt{2}$ and $\ket{-} := (\ket{0}-\ket{1})/\sqrt{2}$.

More interesting is the fact that \textbf{Shifts} is essentially the only nontrivial $3$-qubit UPB in the sense that every UPB in $(\mathbb{C}^2)^{\otimes 3}$ either spans the entire $8$-dimensional space or is equal to \textbf{Shifts} up to permuting the qubits and changing the bases used on each qubit \cite{Bra04} (i.e., replacing the basis $\{\ket{0},\ket{1}\}$ with another basis $\{\ket{a},\ket{\overline{a}}\}$ on any or all qubits, and similarly replacing $\{\ket{+},\ket{-}\}$ by another basis $\{\ket{b},\ket{\overline{b}}\}$ on any or all qubits). In other words, all nontrivial UPBs on $3$ qubits have the same othogonality graph, depicted in Figure~\ref{fig:shifts_og}.
	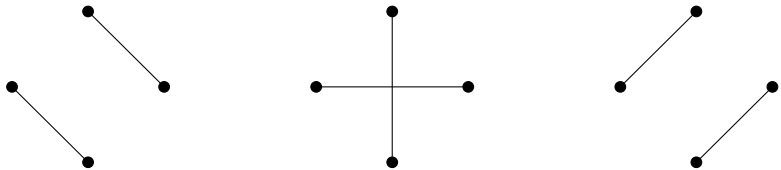
\begin{figure}[htb]
		\centering
		\begin{tikzpicture}[x=1cm, y=1cm, label distance=0cm]
		    \draw[draw=black] (0,1) -- (1,0);
		    \draw[draw=black] (0,-1) -- (-1,0);
	
			\vertex[fill] (v00) at (0,1) []{};
			\vertex[fill] (v01) at (1,0) []{};
			\vertex[fill] (v02) at (0,-1) []{};
			\vertex[fill] (v03) at (-1,0) []{};

		    \draw[draw=black] (4,1) -- (4,-1);
		    \draw[draw=black] (3,0) -- (5,0);
	
			\vertex[fill] (v00) at (4,1) []{};
			\vertex[fill] (v01) at (5,0) []{};
			\vertex[fill] (v02) at (4,-1) []{};
			\vertex[fill] (v03) at (3,0) []{};

		    \draw[draw=black] (8,1) -- (7,0);
		    \draw[draw=black] (9,0) -- (8,-1);
	
			\vertex[fill] (v00) at (8,1) []{};
			\vertex[fill] (v01) at (9,0) []{};
			\vertex[fill] (v02) at (8,-1) []{};
			\vertex[fill] (v03) at (7,0) []{};
		\end{tikzpicture}
		\caption{The orthogonality graph of the \textbf{Shifts} UPB (and every other nontrivial UPB on $3$ qubits).}\label{fig:shifts_og}
	\end{figure}

\subsection{Minimum and Maximum Size}\label{section:minmax}

One of the first questions asked about unextendible product bases was what their possible sizes are. While a full answer to this question is still out of reach, the minimum and maximum size of qubit UPBs is now known. It was shown in \cite{Joh13UPB} that if we define a function $f : \mathbb{N} \rightarrow \mathbb{N}$ by
\begin{align}\label{eq:qubit_min_size}
	f(p) := \begin{cases}
		p + 1 & \text{if $p$ is odd} \\
		p + 2 & \text{if $p = 4$ or $p \equiv 2 (\text{mod } 4)$} \\
		p + 3 & \text{if $p = 8$} \\
		p + 4 & \text{otherwise,}
	\end{cases}
\end{align}
then the smallest $p$-qubit UPB has size $f(p)$.

At the other end of the spectrum, it is straightforward to see that the maximum size of a $p$-qubit UPB is $2^p$, since the standard basis forms a UPB. However, UPBs that span the entire $2^p$-dimensional space are typically not considered to be particularly interesting, so it is natural to instead ask for the maximum size of a \emph{nontrivial} UPB (i.e. one whose size is strictly less than $2^p$). It is straightforward to use the \textbf{Shifts} UPB together with induction and Proposition~\ref{prop:qubit_add_together} to show that there exists a nontrivial $p$-qubit UPB of size $2^p - 4$ for all $p \geq 3$. The following proposition, which is likely known, shows that this is always the largest nontrivial UPB.
\begin{prop}\label{prop:no_large_qubit_upbs}
	Let $p,s \in \mathbb{N}$. There does not exist a $p$-qubit UPB of size $s$ when $2^p - 4 < s < 2^p$.
\end{prop}
\begin{proof}
	Given a UPB $\{\ket{v_1},\ldots,\ket{v_{s}}\} \subset (\mathbb{C}^2)^{\otimes p}$, we can construct the (unnormalized) $p$-qubit mixed quantum state $\rho := I - \sum_{i=1}^{s} \ketbra{v_i}{v_i}$, which has rank $2^p - s$ and has positive partial transpose across any partition of the qubits. Furthermore, $\rho$ is entangled by the range criterion \cite{H97}. If $s = 2^p - 1$ then ${\rm rank}(\rho) = 1$, and it is well-known that pure states with positive partial transpose are necessarily separable, which is a contradiction that shows that no UPB of size $s = 2^p - 1$ exists. Similarly, it was shown in that \cite{CD13} that every multipartite state of rank $2$ or $3$ with positive partial transpose is separable, which shows that no UPB of size $s = 2^p - 2$ or $s = 2^p - 3$ exists.
\end{proof}

\section{Four-Qubit UPBs}\label{section:4qubit}

This section is devoted to describing all of the nontrivial UPBs in $(\mathbb{C}^2)^{\otimes p}$ when $p = 4$. Unlike in the $p = 3$ case, which we saw earlier admits a very simple characterization in terms of the {\bf Shifts} UPB, there are many different four-qubit UPBs. More specifically, we will see here that there are exactly $1446$ inequivalent nontrivial four-qubit UPBs, $1137$ of which arise from combining two $3$-qubit UPBs via Proposition~\ref{prop:qubit_add_together} and $309$ of which are not decomposable in this way. Note that all $4$-qubit UPBs with at most two bases per qubit were found in \cite{SFABCLA13}, however this is the first characterization of \emph{all} $4$-qubit UPBs (including those with three or more bases per qubit).

The process of finding these $4$-qubit UPBs (steps 1 and 2 in Appendix~A) as well as the process of characterizing these UPBs and determining that they are inequivalent (step 3 in Appendix~A) were both done via computer search. It is generally not straightforward to see that a given UPB is indeed unextendible, and it also does not seem to be easy to determine whether or not two given UPBs are equivalent.

Rather than trying to prove that these UPBs are indeed unextendible or are inequivalent as we claim, here we focus instead on summarizing the results, categorizing them as efficiently as possible, and explaining where already-known UPBs fit into this characterization. We sort the different four-qubit UPBs by their size, starting with the minimal $6$-state UPB and working our way up to the maximal (nontrivial) $12$-state UPBs. For a succinct summary of these results, see Table~\ref{tab:4qubit}. An explicit list of all $1446$ inequivalent four-qubit UPBs is available for download from \cite{Joh4QubitUPBs}.

\subsection{Four-Qubit UPBs of 6 States}\label{section:4qubit6}

It was already known that there exists a $6$-state UPB on $4$ qubits, and that this is the minimum size possible \cite{Fen06}:
\begin{align*}
	\big\{ \ket{0000},\ket{0aa1},\ket{10ba},\ket{1a\overline{b}b},\ket{a1\overline{ab}},\ket{\overline{aa}1\overline{a}} \big\}.
\end{align*}
Our computer search showed that this $6$-state UPB is essentially unique---all other UPBs in this case are equivalent to it.

\subsection{Four-Qubit UPBs of 7 States}\label{section:4qubit7}

A $7$-state UPB on $4$ qubits was found via computer search in \cite{AFKKPLA12}:
\begin{align*}
	\big\{ \ket{0000},\ket{0aa1},\ket{0\overline{a}1a},\ket{100b},\ket{1\overline{a}a\overline{b}},\ket{aa10}, \ket{\overline{a}1\overline{aa}} \big\}.
\end{align*}
Once again, our computer search showed that this UPB is essentially unique in the sense that all other $7$-state UPBs on $4$ qubits are equivalent to it.

\subsection{Four-Qubit UPBs of 8 States}\label{section:4qubit8}

This case is much less trivial than the previous two cases. First, note that we can use Proposition~\ref{prop:qubit_add_together} to construct an $8$-state UPB via two copies of the $3$-qubit {\bf Shifts} UPB. In fact, there are many slightly different ways to do this, since we are free to let the bases used in one of the copies of {\bf Shifts} be the same or different from any of the bases used in the other copy of {\bf Shifts}. For example, we can consider the following two UPBs:
\begin{align*}
	\text{UPB}_1 & := \big\{ \ket{000} \otimes \ket{0}, \ket{1{+}{-}} \otimes \ket{0}, \ket{{-}1{+}} \otimes \ket{0}, \ket{{+}{-}1} \otimes \ket{0} \big\} \\
		& \quad \quad \cup \big\{ \ket{000} \otimes \ket{1}, \ket{1{+}{-}} \otimes \ket{1}, \ket{{-}1{+}} \otimes \ket{1}, \ket{{+}{-}1} \otimes \ket{1} \big\} \text{ and} \\
	\text{UPB}_2 & := \big\{ \ket{000} \otimes \ket{0}, \ket{1{+}{-}} \otimes \ket{0}, \ket{{-}1{+}} \otimes \ket{0}, \ket{{+}{-}1} \otimes \ket{0} \big\} \\
		& \quad \quad \cup \big\{ \ket{000} \otimes \ket{1}, \ket{1a\overline{a}} \otimes \ket{1}, \ket{\overline{a}1a} \otimes \ket{1}, \ket{a\overline{a}1} \otimes \ket{1} \big\}.
\end{align*}

It is straightforward to see that UPB$_1$ and UPB$_2$ are inequivalent, since they have different orthogonality graphs. However, they both can be seen as arising from Proposition~\ref{prop:qubit_add_together}: UPB$_1$ arises from following the construction given in its proof exactly, while UPB$_2$ arises from replacing $\{\ket{{+}},\ket{{-}}\}$ by $\{\ket{a},\ket{\overline{a}}\}$ on one copy of {\bf Shifts} before following the construction. Our computer search found that there are exactly $89$ inequivalent UPBs that arise from two copies of {\bf Shifts} and Proposition~\ref{prop:qubit_add_together} in this manner.

Furthermore, there are also $55$ inequivalent UPBs in this case that are really ``new''---they can not be constructed via Proposition~\ref{prop:qubit_add_together} in any way. One of these $55$ UPBs was found in \cite{AFKKPLA12}:
\begin{align*}
	\big\{ \ket{0000}, \ket{1a\overline{a}a}, \ket{a\overline{a}1\overline{a}}, \ket{\overline{a}1ab}, \ket{0a\overline{a}1}, \ket{1a\overline{aa}}, \ket{a1aa}, \ket{\overline{aa}1\overline{b}} \big\},
\end{align*}
and several more were found in \cite{SFABCLA13}.

This gives a total of $144$ inequivalent $8$-state UPBs on $4$ qubits. We note that some (but not all) of these new UPBs can be constructed via the method given in the proof of the upcoming Theorem~\ref{thm:qubit_4k_upbs}.

\subsection{Four-Qubit UPBs of 9 States}\label{section:4qubit9}

Our computer search found that there are exactly $11$ inequivalent $4$-qubit UPBs in this case, which are presented in their entirety in Table~\ref{tab:4qubit}. The following two of these UPBs were found in \cite{AFKKPLA12}:
\begin{align*}
	& \big\{ \ket{0000}, \ket{1a\overline{a}0}, \ket{a\overline{a}10}, \ket{\overline{a}1aa}, \ket{0001}, \ket{01\overline{a}1}, \ket{1\overline{a}0\overline{a}}, \ket{0011}, \ket{1011} \big\} \text{ and} \\
	& \big\{ \ket{0000}, \ket{\overline{a}a1a}, \ket{a1a1}, \ket{\overline{a}11\overline{a}}, \ket{aa\overline{a}1}, \ket{1\overline{aa}a}, \ket{10a\overline{a}}, \ket{\overline{a}10\overline{a}}, \ket{a1a0} \big\},
\end{align*}
while the other $9$ UPBs are new.
\begin{table}
\begin{center}
	\def\arraystretch{1.3}
    \begin{tabular}{@{}l l l@{}}
    \toprule
    size & \# & summary of $4$-qubit UPBs \\ \midrule
    $6$ & $1$ & $\ket{0000},\ket{0aa1},\ket{10ba},\ket{1a\overline{b}b},\ket{a1\overline{ab}},\ket{\overline{aa}1\overline{a}} \quad$ (unique) \\\hline
    $7$ & $1$ & $\ket{0000},\ket{0aa1},\ket{0\overline{a}1a},\ket{100b},\ket{1\overline{a}a\overline{b}},\ket{aa10}, \ket{\overline{a}1\overline{aa}} \quad$ (unique) \\\hline
    $8$ & $144$ & $89$ are of the form $({\bf Shifts} \otimes \ket{0}) \cup ({\bf Shifts} \otimes \ket{1})$ \\
     & & plus $55$ others, such as: \\
     & & $\ket{0000}, \ket{1a\overline{a}a}, \ket{a\overline{a}1\overline{a}}, \ket{\overline{a}1ab}, \ket{0a\overline{a}1}, \ket{1a\overline{aa}}, \ket{a1aa}, \ket{\overline{aa}1\overline{b}}$ \\\hline
    $9$ & $11$ & $\ket{0000}, \ket{1a\overline{a}0}, \ket{a\overline{a}10}, \ket{\overline{a}1aa}, \ket{0001}, \ket{01\overline{a}1}, \ket{1\overline{a}0\overline{a}}, \ket{0011}, \ket{1011}$ \\
    & & $\ket{0000}, \ket{\overline{a}a1a}, \ket{a1a1}, \ket{\overline{a}11\overline{a}}, \ket{aa\overline{a}1}, \ket{1\overline{aa}a}, \ket{10a\overline{a}}, \ket{\overline{a}10\overline{a}}, \ket{a1a0}$ \\
    & & $\ket{0000}, \ket{0001}, \ket{0010}, \ket{010a}, \ket{1aaa}, \ket{100\overline{a}}, \ket{11\overline{a}0}, \ket{a1a\overline{a}}, \ket{\overline{aa}11}$ \\
    & & $\ket{0000}, \ket{0001}, \ket{001a}, \ket{010b}, \ket{1aab}, \ket{100\overline{b}}, \ket{11\overline{a}a}, \ket{a1a\overline{b}}, \ket{\overline{aa}1\overline{a}}$ \\
    & & $\ket{0000}, \ket{0001}, \ket{001a}, \ket{01aa}, \ket{1aab}, \ket{10\overline{a}a}, \ket{110\overline{b}}, \ket{a1\overline{a}b}, \ket{\overline{aa}1\overline{a}}$ \\
    & & $\ket{0000}, \ket{0001}, \ket{01aa}, \ket{01\overline{a}a}, \ket{1a0a}, \ket{1\overline{aa}b}, \ket{a01\overline{b}}, \ket{a1a\overline{a}}, \ket{\overline{a}a1\overline{a}}$ \\
    & & $\ket{0000}, \ket{0001}, \ket{01aa}, \ket{01\overline{a}a}, \ket{1a0a}, \ket{1\overline{a}bb}, \ket{a01\overline{b}}, \ket{a1\overline{ba}}, \ket{\overline{a}a1\overline{a}}$ \\
    & & $\ket{0000}, \ket{0001}, \ket{001a}, \ket{01aa}, \ket{1a0b}, \ket{101a}, \ket{1\overline{a}a\overline{a}}, \ket{aa1\overline{a}}, \ket{\overline{a}1\overline{ab}}$ \\
    & & $\ket{0000}, \ket{001a}, \ket{001\overline{a}}, \ket{01a0}, \ket{1aaa}, \ket{10\overline{a}0}, \ket{111\overline{a}}, \ket{a1\overline{a}a}, \ket{\overline{aa}01}$ \\
    & & $\ket{0000}, \ket{001a}, \ket{001\overline{a}}, \ket{01a0}, \ket{1a1\overline{a}}, \ket{1000}, \ket{1\overline{a}a1}, \ket{aa01}, \ket{\overline{a}1\overline{a}a}$ \\
    & & $\ket{0000}, \ket{01aa}, \ket{0a1\overline{a}}, \ket{1110}, \ket{1a0a}, \ket{10a\overline{a}}, \ket{a01a}, \ket{a10\overline{a}}, \ket{\overline{aaa}1}$ \\ \hline
    $10$ & $80$ & $\ket{0000}, \ket{1a\overline{a}0}, \ket{a\overline{a}10}, \ket{\overline{a}1aa}, \ket{0001}, \ket{0011}, \ket{1001}, \ket{1011}, \ket{010\overline{a}}, \ket{11\overline{a}1}$ \\
    & & plus $79$ others \\\hline
    $11$ & $0$ & \\\hline
    $12$ & $1209$ & $1048$ are of the form $(\mathbf{Shifts} \otimes \ket{0}) \cup (\mathbf{B}_i \otimes \ket{1})$ for some $i$ (see Table~\ref{tab:3qubitproductbases}) \\
    & & plus $161$ others, such as: \\
    & & $\ket{0000}, \ket{\overline{a}aa1}, \ket{a11a}, \ket{\overline{a}1\overline{ab}}, \ket{1000}, \ket{a001},$ \\
    & & $\quad \ket{a10\overline{a}}, \ket{a010}, \ket{a011}, \ket{a11\overline{a}}, \ket{\overline{aa}1b}, \ket{a10a}$ \\\bottomrule
    \end{tabular}
    \caption{A summary of the $1446$ inequivalent $4$-qubit UPBs. A complete list can be found at \cite{Joh4QubitUPBs}.}\label{tab:4qubit}
\end{center}
\end{table}

\subsection{Four-Qubit UPBs of 10 States}\label{section:4qubit10}

Once again, it was already known that a $10$-state UPB exists, as one was found in \cite{AFKKPLA12}:
\begin{align*}
	\big\{ \ket{0000}, \ket{1a\overline{a}0}, \ket{a\overline{a}10}, \ket{\overline{a}1aa}, \ket{0001}, \ket{0011}, \ket{1001}, \ket{1011}, \ket{010\overline{a}}, \ket{11\overline{a}1} \big\},
\end{align*}
and one more was found in \cite{SFABCLA13}. We have found that there are $78$ more inequivalent UPBs, for a total of $80$.

\subsection{Four-Qubit UPBs of 11 States}\label{section:4qubit11}

Our computer search showed that there does not exist an $11$-state UPB on $4$ qubits. This case is interesting for at least two reasons. First, it shows that UPBs are not ``continuous'' in the sense that there can be nontrivial UPBs of sizes $s-1$ and $s+1$ in a given space, yet no UPB of size $s$ (we will see in Section~\ref{section:manyqubit} that UPBs are also not ``continuous'' in this sense when the number of qubits is odd and at least $5$).

Second, this is currently the only case where it is known that no UPB exists via means other than an explicit (human-readable) proof. This raises the question of whether or not there is a ``simple'' proof of the fact that there is no $11$-state UPB on $4$ qubits. More generally, it would be interesting to determine whether or not there exists a $p$-qubit UPB of size $2^p - 5$ when $p \geq 5$---we will see in Section~\ref{section:manyqubit} that this is the only unsolved case that is near the $2^p - 4$ upper bound.

\subsection{Four-Qubit UPBs of 12 States}\label{section:4qubit12}

The vast majority of $4$-qubit UPBs arise in this case. Similar to the $8$-state case considered in Section~\ref{section:4qubit8}, we can construct many $12$-state UPBs by using Proposition~\ref{prop:qubit_add_together} to combine $3$-qubit UPBs of size $4$ and $8$ (i.e., {\bf Shifts} and a full $3$-qubit product basis). However, things are more complicated in this case, as there are $17$ inequivalent $3$-qubit product bases of size $8$ that can be used in Proposition~\ref{prop:qubit_add_together}.
\begin{table}
\begin{center}
	\def\arraystretch{1.3}
    \begin{tabular}{@{}l l l@{}}
    \toprule
    $\quad \quad \ \ \ 3$-qubit product basis & & $4$-qubit UPBs \\ \midrule
    $\mathbf{B}_1 := \{\ket{000},\ket{001},\ket{010},\ket{011},\ket{100},\ket{101},\ket{110},\ket{111}\}$ & & $5$ \\
    $\mathbf{B}_{2} := \{\ket{000}, \ket{001}, \ket{010}, \ket{011}, \ket{100}, \ket{101}, \ket{11a}, \ket{11\overline{a}}\}$ & & $32$ \\
    $\mathbf{B}_{3} := \{\ket{000}, \ket{001}, \ket{010}, \ket{011}, \ket{10a}, \ket{10\overline{a}}, \ket{11b}, \ket{11\overline{b}}\}$ & & $47$ \\
    $\mathbf{B}_{4} := \{\ket{000}, \ket{001}, \ket{010}, \ket{011}, \ket{10a}, \ket{1a\overline{a}}, \ket{11a}, \ket{1\overline{aa}}\}$ & & $99$ \\
    $\mathbf{B}_{5} := \{\ket{000}, \ket{001}, \ket{010}, \ket{011}, \ket{1aa}, \ket{1a\overline{a}}, \ket{1\overline{a}a}, \ket{1\overline{aa}}\}$ & & $25$ \\
    $\mathbf{B}_{6} := \{\ket{000}, \ket{001}, \ket{010}, \ket{011}, \ket{1aa}, \ket{1a\overline{a}}, \ket{1\overline{a}b}, \ket{1\overline{ab}}\}$ & & $93$ \\
    $\mathbf{B}_7 := \{\ket{000}, \ket{001}, \ket{01a}, \ket{01\overline{a}}, \ket{100}, \ket{101}, \ket{11a}, \ket{11\overline{a}}\}$ & & $18$ \\
    $\mathbf{B}_{8} := \{\ket{000}, \ket{001}, \ket{01a}, \ket{01\overline{a}}, \ket{100}, \ket{1a1}, \ket{110}, \ket{1\overline{a}1}\}$ & & $85$ \\
    $\mathbf{B}_{9} := \{\ket{000}, \ket{001}, \ket{01a}, \ket{01\overline{a}}, \ket{10a}, \ket{10\overline{a}}, \ket{110}, \ket{111}\}$ & & $13$ \\
    $\mathbf{B}_{10} := \{\ket{000}, \ket{001}, \ket{01a}, \ket{01\overline{a}}, \ket{10b}, \ket{10\overline{b}}, \ket{110}, \ket{111}\}$ & & $34$ \\
    $\mathbf{B}_{11} := \{\ket{000}, \ket{001}, \ket{01a}, \ket{01\overline{a}}, \ket{10b}, \ket{10\overline{b}}, \ket{11c}, \ket{11\overline{c}}\}$ & & $26$ \\
    $\mathbf{B}_{12} := \{\ket{000}, \ket{001}, \ket{01a}, \ket{01\overline{a}}, \ket{10b}, \ket{1a\overline{b}}, \ket{11b}, \ket{1\overline{ab}}\}$ & & $143$ \\
    $\mathbf{B}_{13} := \{\ket{000}, \ket{001}, \ket{01a}, \ket{01\overline{a}}, \ket{1a0}, \ket{1a1}, \ket{1\overline{a}a}, \ket{1\overline{aa}}\}$ & & $51$ \\
    $\mathbf{B}_{14} := \{\ket{000}, \ket{001}, \ket{01a}, \ket{01\overline{a}}, \ket{1a0}, \ket{1a1}, \ket{1\overline{a}b}, \ket{1\overline{ab}}\}$ & & $142$ \\
    $\mathbf{B}_{15} := \{\ket{000}, \ket{001}, \ket{01a}, \ket{01\overline{a}}, \ket{1ab}, \ket{1\overline{a}b}, \ket{1b\overline{b}}, \ket{1\overline{bb}}\}$ & & $75$ \\
    $\mathbf{B}_{16} := \{\ket{000}, \ket{001}, \ket{01a}, \ket{01\overline{a}}, \ket{1ab}, \ket{1a\overline{b}}, \ket{1\overline{a}c}, \ket{1\overline{ac}}\}$ & & $81$ \\
    $\mathbf{B}_{17} := \{\ket{000}, \ket{01a}, \ket{01\overline{a}}, \ket{1a0}, \ket{1\overline{a}0}, \ket{a01}, \ket{\overline{a}01}, \ket{111}\}$ & & $79$ \\\cmidrule{3-3}
    \multicolumn{1}{r}{Total:} & & $1048$ \\\bottomrule
    \end{tabular}
    \caption{A summary of the $17$ inequivalent $3$-qubit orthogonal product bases. The table also gives the number of inequivalent $12$-state $4$-qubit UPBs that these product bases give rise to by being combined with {\bf Shifts} via Proposition~\ref{prop:qubit_add_together}.}\label{tab:3qubitproductbases}
\end{center}
\end{table}

These $17$ product bases as well as the number of inequivalent $12$-state $4$-qubit UPBs that they give rise to via Proposition~\ref{prop:qubit_add_together} are given in Table~\ref{tab:3qubitproductbases}. For example, there are $5$ inequivalent $4$-qubit UPBs of the form $(\mathbf{Shifts} \otimes \ket{0}) \cup (\mathbf{B}_1 \otimes \ket{1})$, where $\mathbf{B}_1$ is the standard basis of $(\mathbb{C}^2)^{\otimes 3}$. A total of $1048$ inequivalent $4$-qubit UPBs can be constructed in this way by using the $17$ different $3$-qubit product bases.

Furthermore, there are also $161$ inequivalent UPBs in this case that can not be constructed via Proposition~\ref{prop:qubit_add_together}, for a total of $1209$ inequivalent $12$-state UPBs on $4$ qubits. To the best of our knowledge, only two of these $161$ UPBs have been found before \cite{AFKKPLA12}:
\begin{align*}
	& \big\{ \ket{0000}, \ket{\overline{a}aa1}, \ket{a11a}, \ket{\overline{a}1\overline{ab}}, \ket{1000}, \ket{a001}, \ket{a10\overline{a}}, \ket{a010}, \ket{a011}, \ket{a11\overline{a}}, \ket{\overline{aa}1b}, \ket{a10a} \big\},\\
	& \big\{ \ket{0000}, \ket{1aaa}, \ket{a\overline{a}1b}, \ket{10\overline{ab}}, \ket{0ab1}, \ket{01\overline{bb}}, \ket{1\overline{a}0b}, \ket{1aa\overline{a}}, \ket{1\overline{a}a\overline{b}}, \ket{1a\overline{a}b}, \ket{11\overline{ab}}, \ket{\overline{aa}1b} \big\}.
\end{align*}

\section{Five- and Six-Qubit UPBs}\label{section:5qubit}

In $(\mathbb{C}^2)^{\otimes 5}$, the minimum and maximum sizes of UPBs are well-known to be $6$ and $28$, respectively, but otherwise very little is known. The only UPBs in this case that have appeared in the past that we are aware of are the {\bf GenShifts} UPB of size $6$ \cite{DMSST03} and the UPBs of sizes $12$--$26$ and $28$ that can be created by combining two $4$-qubit UPBs via Proposition~\ref{prop:qubit_add_together}. This leaves UPBs of size $7, 8, 9, 10, 11,$ and $27$ unaccounted for.

Our computer search has shown that there does not exist a $5$-qubit UPB of size $7$, but there do exist UPBs of size $8$, $9$, and $10$:
\begin{align*}
	& \big\{ \ket{00000},\ket{001aa},\ket{aaa1\overline{a}},\ket{aa\overline{aa}1},\ket{1\overline{a}bbb}, \ket{1\overline{ab}cc},\ket{\overline{a}1c\overline{bc}},\ket{\overline{a}1\overline{ccb}} \big\}, \\
	& \big\{ \ket{00000},\ket{0aa01},\ket{0b1a0},\ket{1abaa},\ket{1b\overline{b}bb},\ket{a\overline{b}a1\overline{a}},\ket{a1\overline{aab}},\ket{\overline{aa}0\overline{b}1},\ket{\overline{ab}11\overline{a}} \big\} \text{ and} \\
	& \big\{ \ket{00000},\ket{0a001},\ket{0b1aa},\ket{10abb},\ket{1abb\overline{b}},\ket{ab1\overline{ba}},\ket{a\overline{bb}1\overline{b}},\ket{a\overline{ba}1b},\ket{\overline{a}1\overline{ba}0},\ket{\overline{aaaa}1} \big\}.
\end{align*}
In fact, we have completely characterized qubit UPBs of $8$ or fewer states (on any number of qubits), which are available for download from \cite{Joh5QubitUPBs}. We have not been able to prove or disprove the existence of $5$-qubit UPBs of size $11$ or $27$, as they are beyond our computational capabilities. We leave them as open problems (see Table~\ref{tab:qubit_summary}).

The case of $6$-qubit UPBs is quite similar, with sizes $9$, $10$, $11$, $13$, and $59$ unknown. We have found a $6$-qubit UPB of $9$ states, leaving four cases still unsolved:
\begin{align*}
	\big\{ \ket{000000}, \ket{0aaaa1}, \ket{1aabaa}, \ket{1bb\overline{b}bb}, \ket{a\overline{a}b1\overline{b}c}, \ket{ab\overline{ba}1\overline{a}}, \ket{\overline{a}1\overline{a}cc\overline{b}}, \ket{b\overline{ab}a\overline{c}1}, \ket{\overline{bb}1\overline{cac}} \big\}
\end{align*}

\section{Many-Qubit UPBs}\label{section:manyqubit}

We now turn our attention to the construction of UPBs on an arbitrary number of qubits. We already saw that $p$-qubit UPBs are not ``continuous'' when $p = 4$ (where there are UPBs of size $10$ and $12$, but none of size $11$) or $p = 5$ (where there are UPBs of size $6$ and $8$, but none of size $7$). Our first result shows that the same is true whenever $p \geq 5$ is odd.
\begin{prop}\label{prop:qubit_no_p2}
	If $p$ is odd then there does not exist a UPB in $(\bb{C}^{2})^{\otimes p}$ consisting of exactly $p+2$ states.
\end{prop}
\begin{proof}
	Suppose for a contradiction that such a UPB exists. We first note that there can not be a set of $3$ or more states of the UPB that are equal to each other on a given qubit, or else unextendibility is immediately violated, since we could construct a product state orthogonal to all $3$ of those states on that qubit and orthogonal to one other state on each of the other $p-1$ qubits, for a total of all $3 + (p-1) = p+2$ states.
	
	Additionally, since $p+2$ is odd, \cite[Lemma~2]{Joh13UPB} implies that on every party there is a pair of two states of the UPB that are equal to each other (and furthermore, the number of such pairs must be odd).
	
	We now argue that there must be exactly one such pair on every party. To see this, suppose for a contradiction that there are three or more pairs of states that are equal to each other on some qubit (which we assume without loss of generality is the first qubit). On the second qubit, there is at least one pair of states that are equal to each other, and this pair contains no vertices in common with at least one of the pairs of states that are equal to each other on the first qubit. Thus we can find a product state that is orthogonal to $2$ states on the second qubit, $2$ more states on the first qubit, and $1$ more state on each of the remaining $p-2$ qubits, for a total of all $2 + 2 + (p-2) = p+2$ states. It follows that the UPB is extendible, so in fact there can only be one pair of equal states on each party, as in Figure~\ref{fig:prop_p2_proof}.
	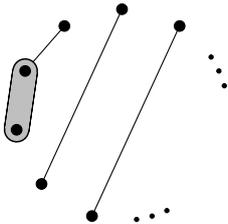
\begin{figure}[htb]
		\centering
		\begin{tikzpicture}[x=1.4cm, y=1.4cm, label distance=0cm]
		    \draw[draw=black] (-0.910,0.415) -- (-0.541,0.841);
		    \draw[draw=black] (0,1) -- (-0.756,-0.655);
		    \draw[draw=black] (0.541,0.841) -- (-0.282,-0.959);
	
		    \filldraw[fill=lightgray,line width=0.34cm,line join=round,draw=black] (-0.910,0.415) -- (-0.990,-0.142) -- (-0.9500,0.1365) -- cycle;
	    	\filldraw[fill=lightgray,line width=0.30cm,line join=round,draw=lightgray] (-0.910,0.415) -- (-0.990,-0.142) -- (-0.9500,0.1365) -- cycle;
	
			\vertex[fill] (v00) at (0,1) []{};
			\vertex[fill] (v01) at (-0.541,0.841) []{};
			\vertex[fill] (v02) at (-0.910,0.415) []{};
			\vertex[fill] (v03) at (-0.990,-0.142) []{};
			\vertex[fill] (v04) at (-0.756,-0.655) []{};
			\vertex[fill] (v05) at (-0.282,-0.959) []{};
			\vertex[fill] (v010) at (0.541,0.841) []{};
	
			\draw[thin,fill] (0.837,0.547) circle(0.02);
			\draw[thin,fill] (0.910,0.415) circle(0.02);
			\draw[thin,fill] (0.961,0.275) circle(0.02);
			
			\draw[thin,fill] (0.137,-0.991) circle(0.02);
			\draw[thin,fill] (0.282,-0.959) circle(0.02);
			\draw[thin,fill] (0.421,-0.907) circle(0.02);
		\end{tikzpicture}
		\caption{What the orthogonality graph of every party of a $(p+2)$-state UPB on $p$ qubits would have to look like when $p$ is odd (up to repositioning vertices).}\label{fig:prop_p2_proof}
	\end{figure}

	However, it then directly follows that the orthogonality graph of each qubit can contain no more than $(p+1)/2$ edges, so the orthogonality graph of all $p$ qubits contains no more than $p(p+1)/2$ edges. However, in order for the $p+2$ states of the UPB to be mutually orthogonal, there would have to be at least $(p+1)(p+2)/2 > p(p+1)/2$ edges present in the orthogonality graph, so it follows that some of these states are not orthogonal on any qubit and thus do not form a UPB.
\end{proof}

We have now seen examples that demonstrate that there are sizes near the minimal size $f(p)$ (defined in Equation~\eqref{eq:qubit_min_size}) for which no $p$-qubit UPB exists, and similarly there are sizes near the maximal size $2^p - 4$ for which no $p$-qubit UPB exists (e.g., there is no $4$-qubit UPB of size $11 = 2^p - 5$). We now show that these are essentially the only possible cases where qubit UPBs do not exist---there exist UPBs of all sizes that are sufficiently far in between the minimal and maximal sizes.
\begin{thm}\label{thm:many_qubit_sizes}
	If $p \geq 7$ then there exists a $p$-qubit UPB of size $s$ whenever
	\begin{align*}
		\frac{p^2 + 3p - 30}{2} \leq s \leq 2^p - 6.
	\end{align*}
\end{thm}
Before proving this result, we note that we actually prove the slightly better lower bound $\sum_{k=4}^{p-1} f(k)$. However, these two lower bounds never differ by more than $2$ (which is proved in Appendix~B), so we prefer the present statement of the result with the lower bound $(p^2 + 3p - 30)/2$, which is much easier to work with.
\begin{proof}
	We prove the result via Proposition~\ref{prop:qubit_add_together} and induction on $p$. As indicated above, we actually prove the slightly stronger statement that such a UPB exists whenever $\sum_{k=4}^{p-1} f(k) \leq s \leq 2^p - 6$.
	
	For the base case $p = 7$, recall from Section~\ref{section:5qubit} that there exist $6$-qubit UPBs of sizes 8, 9, 12, 14--58, 60, and 64. It follows from Proposition~\ref{prop:qubit_add_together} that there exist $7$-qubit UPBs of sizes 16--18, 20--122, 124, and 128. We thus see that there is a $7$-qubit UPB of any size from $\sum_{k=4}^{6} f(k) = 6 + 6 + 8 = 20$ to $2^7 - 6 = 122$, as desired.
	
	For the inductive step, fix $p$ and define the following four intervals of positive integers:
	\begin{align*}
		I_0^p & := \big[\sum_{k=4}^{p-1} f(k), 2^p - 6\big] &
		I_1^p & := \big[\sum_{k=4}^{p} f(k), f(p) + (2^p - 6)\big] \\
		I_2^p & := \big[ 2\sum_{k=4}^{p-1} f(k), 2^{p+1} - 12 \big] &
		I_3^p & := \big[ 2^p + \sum_{k=4}^{p-1} f(k), 2^{p+1} - 6 \big].
	\end{align*}
	Assume that there exist $p$-qubit UPBs of all sizes $s \in I_0^p$ and notice that the intervals we have defined satisfy the following relationships:
	\begin{align}\label{eq:interval_rel}
		I_1^p & = f(p) + I_0^p, & I_2^p & = I_0^p + I_0^p, & I_3^p & = 2^p + I_0^p.
	\end{align}
	
	Equations~\eqref{eq:interval_rel} hint at the remainder of the proof. We can combine a minimal $p$-qubit UPB of size $f(p)$ with one of the $p$-qubit UPBs with size in $I_0^p$ via Proposition~\ref{prop:qubit_add_together} to obtain a $(p+1)$-qubit UPB of any size in $I_1^p$. Similarly, we can combine two $p$-qubit UPBs with sizes in $I_0^p$ to obtain $(p+1)$-qubit UPBs of any size in $I_2^p$. Finally, we can combine a $p$-qubit UPB of size $2^p$ (e.g., the standard basis) with one of the $p$-qubit UPBs with size in $I_0^p$ to obtain a $(p+1)$-qubit UPB of any size in $I_3^p$.

	In order to complete the inductive step and the proof, it suffices to show that $I_1^p \cup I_2^p \cup I_3^p = I_0^{p+1}$. It is clear that the minimal values of $I_1^p$ and $I_0^{p+1}$ coincide, as do the maximal values of $I_3^p$ and $I_0^{p+1}$, so it is enough to show that $I_2^p$ overlaps with each of $I_1^p$ and $I_3^p$.
	
	In order to show that $I_1^p$ and $I_2^p$ overlap, we must show that $2\sum_{k=4}^{p-1} f(k) \leq f(p) + (2^p - 6)$. This inequality can be seen from noting that $f(k) \leq k+4$, so
	\begin{align*}
		2\sum_{k=4}^{p-1} f(k) \leq 2\sum_{k=4}^{p-1} (k+4) = p^2 + 7p - 44 \leq 2^p - 6 \leq f(p) + (2^p - 6),
	\end{align*}
	where we note that the second-to-last inequality can easily be verified by typical methods from calculus.
	
	In order to show that $I_2^p$ and $I_3^p$ overlap, we must show that $2^p + \sum_{k=4}^{p-1} f(k) \leq 2^{p+1} - 12$. Similar to before, this inequality follows straightforwardly:
	\begin{align*}
		2^p + \sum_{k=4}^{p-1} f(k) \leq 2^p + \sum_{k=4}^{p-1} (k+4) = 2^p + (p^2 + 7p - 44)/2 \leq 2^{p+1} - 12,
	\end{align*}
	where the final inequality once again can be verified by straightforward calculus. It follows that $I_1^p \cup I_2^p \cup I_3^p = I_0^{p+1}$, as desired, which completes the proof.
\end{proof}

As an immediate corollary of Theorem~\ref{thm:many_qubit_sizes}, we note that as $p \rightarrow \infty$, almost all cardinalities in the interval $[1,2^p]$ are attainable as the size of a $p$-qubit UPB. This fact can be seen by noting that the proportion of attainable cardinalities is at least
\begin{align*}
	\frac{(2^p - 6) - (p^2 + 3p - 30)/2 + 1}{2^p},
\end{align*}
which tends to $1$ as $p \rightarrow \infty$.

On the other hand, Theorem~\ref{thm:many_qubit_sizes} does not place a very good bound on how large of a ``gap'' $g_p$ there can be such that there exist nontrivial $p$-qubit UPBs of size $s$ and $s + g_p + 1$, but no $p$-qubit UPB of any intermediate size $s+1,\ldots,s+g_p$. Indeed, Theorem~\ref{thm:many_qubit_sizes} does not even guarantee that the maximal value of $g_p$ stays bounded as $p \rightarrow \infty$, since the difference between $(p^2 + 3p - 30)/2$ and $f(p)$ tends to infinity as $p$ does.

We now show that there is indeed an absolute upper bound on how large of a ``gap'' in qubit UPB sizes there can be: $g_p \leq 7$ regardless of $p$, and if $p \not\equiv 1 \, (\text{mod } 4)$ then $g_p \leq 3$. The construction of UPBs presented in the proof of the following theorem generalizes {\bf Shifts} as well as the $5$-qubit UPB of size $8$ that was presented in Section~\ref{section:5qubit}. This result also generalizes \cite[Lemma~4]{Joh13UPB}.
\begin{thm}\label{thm:qubit_4k_upbs}
	Let $p,s \in \mathbb{N}$ be such that $p+1 \leq s \leq 2^p$ and $s$ is a multiple of $4$. Then there is a UPB in $(\bb{C}^{2})^{\otimes p}$ of cardinality $s$, with the possible exception of the case when $p \equiv 1 \, (\text{mod } 4)$ and $s = 2p + 2$.
\end{thm}
\begin{proof}
	We note that it suffices to construct a UPB in the case when $p+1 \leq s \leq 2p$ since the case when $s \geq 2p + 1$ follows directly from Proposition~\ref{prop:qubit_add_together} and induction. For example, when $p \equiv 0 \, (\text{mod } 4)$, we know (by inductive hypothesis) that there are $(p-1)$-qubit UPBs of any size in the set $\{p,p+4,\ldots,2^{p-1}-4,2^{p-1}\}$, and combining these UPBs via Proposition~\ref{prop:qubit_add_together} gives $p$-qubit UPBs of any size in $\{2p,2p+4,\ldots,2^p-4,2^p\}$. The cases when $p \equiv 1,2,3 \, (\text{mod } 4)$ are similar.
	
	We now focus on constructing a UPB in the $s = 2p$ case, and we will generalize this construction to smaller values of $s$ later. Define the integer $k := s/4$. To construct the orthogonality graph of the desired UPB, begin by letting the orthogonality graph on one of the parties be such that every vertex is connected to exactly one other vertex (as in the top graph of Figure~\ref{fig:4k_upb}). On each of the remaining parties, have each of these pairs of states be equal to each other.
	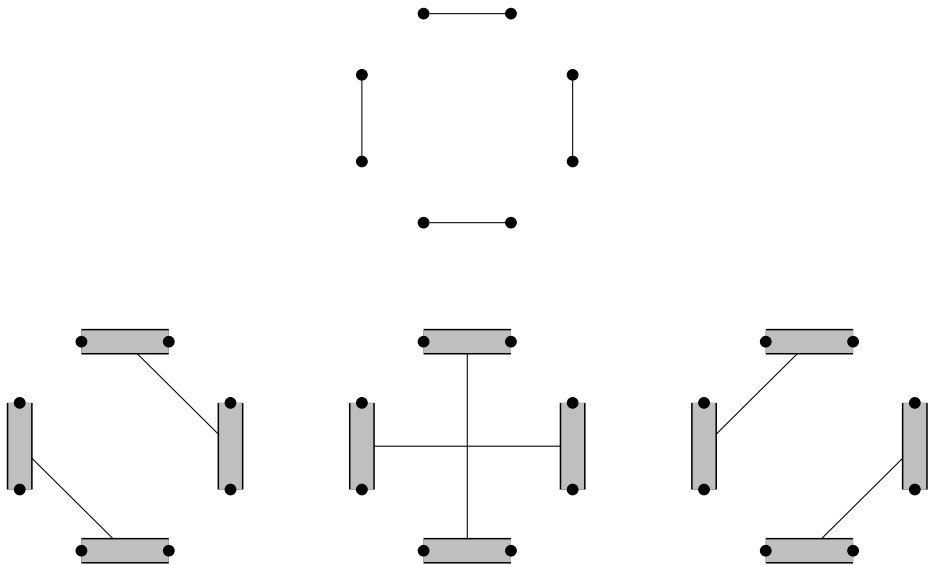
\begin{figure}[htb]
		\centering
		\begin{tikzpicture}[x=1.5cm, y=1.5cm, label distance=0cm]
			\vertex[fill] (v0) at (-0.383,0.824) []{};
			\vertex[fill] (v1) at (-0.924,0.283) []{};
			\vertex[fill] (v2) at (-0.924,-0.483) []{};
			\vertex[fill] (v3) at (-0.383,-1.024) []{};
			\vertex[fill] (v4) at (0.383,-1.024) []{};
			\vertex[fill] (v5) at (0.924,-0.483) []{};
			\vertex[fill] (v6) at (0.924,0.283) []{};
			\vertex[fill] (v7) at (0.383,0.824) []{};
			
		    \draw[draw=black] (0,-2.076) -- (0,-3.924);
		    \draw[draw=black] (-0.924,-3) -- (0.924,-3);

			\filldraw[fill=lightgray,line width=0.34cm,line join=round,draw=black] (-0.924,-2.617) -- (-0.924,-3.383) -- cycle;
		    \filldraw[fill=lightgray,line width=0.30cm,line join=round,draw=lightgray] (-0.924,-2.617) -- (-0.924,-3.383) -- cycle;
			\filldraw[fill=lightgray,line width=0.34cm,line join=round,draw=black] (-0.383,-3.924) -- (0.383,-3.924) -- cycle;
		    \filldraw[fill=lightgray,line width=0.30cm,line join=round,draw=lightgray] (-0.383,-3.924) -- (0.383,-3.924) -- cycle;
			\filldraw[fill=lightgray,line width=0.34cm,line join=round,draw=black] (0.924,-3.383) -- (0.924,-2.617) -- cycle;
		    \filldraw[fill=lightgray,line width=0.30cm,line join=round,draw=lightgray] (0.924,-3.383) -- (0.924,-2.617) -- cycle;
			\filldraw[fill=lightgray,line width=0.34cm,line join=round,draw=black] (-0.383,-2.076) -- (0.383,-2.076) -- cycle;
			\filldraw[fill=lightgray,line width=0.30cm,line join=round,draw=lightgray] (-0.383,-2.076) -- (0.383,-2.076) -- cycle;
			
			\vertex[fill] (w0) at (-0.383,-2.076) []{};
			\vertex[fill] (w1) at (-0.924,-2.617) []{};
			\vertex[fill] (w2) at (-0.924,-3.383) []{};
			\vertex[fill] (w3) at (-0.383,-3.924) []{};
			\vertex[fill] (w4) at (0.383,-3.924) []{};
			\vertex[fill] (w5) at (0.924,-3.383) []{};
			\vertex[fill] (w6) at (0.924,-2.617) []{};
			\vertex[fill] (w7) at (0.383,-2.076) []{};
				
		    \draw[draw=black] (-3,-2.076) -- (-2.076,-3);
		    \draw[draw=black] (-3.924,-3) -- (-3,-3.924);

			\filldraw[fill=lightgray,line width=0.34cm,line join=round,draw=black] (-3.924,-2.617) -- (-3.924,-3.383) -- cycle;
		    \filldraw[fill=lightgray,line width=0.30cm,line join=round,draw=lightgray] (-3.924,-2.617) -- (-3.924,-3.383) -- cycle;
			\filldraw[fill=lightgray,line width=0.34cm,line join=round,draw=black] (-3.383,-3.924) -- (-2.617,-3.924) -- cycle;
		    \filldraw[fill=lightgray,line width=0.30cm,line join=round,draw=lightgray] (-3.383,-3.924) -- (-2.617,-3.924) -- cycle;
			\filldraw[fill=lightgray,line width=0.34cm,line join=round,draw=black] (-2.076,-3.383) -- (-2.076,-2.617) -- cycle;
		    \filldraw[fill=lightgray,line width=0.30cm,line join=round,draw=lightgray] (-2.076,-3.383) -- (-2.076,-2.617) -- cycle;
			\filldraw[fill=lightgray,line width=0.34cm,line join=round,draw=black] (-3.383,-2.076) -- (-2.617,-2.076) -- cycle;
			\filldraw[fill=lightgray,line width=0.30cm,line join=round,draw=lightgray] (-3.383,-2.076) -- (-2.617,-2.076) -- cycle;
			
			\vertex[fill] (x0) at (-3.383,-2.076) []{};
			\vertex[fill] (x1) at (-3.924,-2.617) []{};
			\vertex[fill] (x2) at (-3.924,-3.383) []{};
			\vertex[fill] (x3) at (-3.383,-3.924) []{};
			\vertex[fill] (x4) at (-2.617,-3.924) []{};
			\vertex[fill] (x5) at (-2.076,-3.383) []{};
			\vertex[fill] (x6) at (-2.076,-2.617) []{};
			\vertex[fill] (x7) at (-2.617,-2.076) []{};
				
		    \draw[draw=black] (3,-2.076) -- (2.076,-3);
		    \draw[draw=black] (3.924,-3) -- (3,-3.924);

			\filldraw[fill=lightgray,line width=0.34cm,line join=round,draw=black] (3.924,-2.617) -- (3.924,-3.383) -- cycle;
		    \filldraw[fill=lightgray,line width=0.30cm,line join=round,draw=lightgray] (3.924,-2.617) -- (3.924,-3.383) -- cycle;
			\filldraw[fill=lightgray,line width=0.34cm,line join=round,draw=black] (3.383,-3.924) -- (2.617,-3.924) -- cycle;
		    \filldraw[fill=lightgray,line width=0.30cm,line join=round,draw=lightgray] (3.383,-3.924) -- (2.617,-3.924) -- cycle;
			\filldraw[fill=lightgray,line width=0.34cm,line join=round,draw=black] (2.076,-3.383) -- (2.076,-2.617) -- cycle;
		    \filldraw[fill=lightgray,line width=0.30cm,line join=round,draw=lightgray] (2.076,-3.383) -- (2.076,-2.617) -- cycle;
			\filldraw[fill=lightgray,line width=0.34cm,line join=round,draw=black] (3.383,-2.076) -- (2.617,-2.076) -- cycle;
			\filldraw[fill=lightgray,line width=0.30cm,line join=round,draw=lightgray] (3.383,-2.076) -- (2.617,-2.076) -- cycle;
			
			\vertex[fill] (y0) at (3.383,-2.076) []{};
			\vertex[fill] (y1) at (3.924,-2.617) []{};
			\vertex[fill] (y2) at (3.924,-3.383) []{};
			\vertex[fill] (y3) at (3.383,-3.924) []{};
			\vertex[fill] (y4) at (2.617,-3.924) []{};
			\vertex[fill] (y5) at (2.076,-3.383) []{};
			\vertex[fill] (y6) at (2.076,-2.617) []{};
			\vertex[fill] (y7) at (2.617,-2.076) []{};

			\path 
				(v0) edge (v7)
				(v2) edge (v1)
				(v4) edge (v3)
				(v6) edge (v5)
			;
		\end{tikzpicture}
		\caption{The orthogonality graph constructed in the proof of Theorem~\ref{thm:qubit_4k_upbs} in the $p = 4, s = 8$ case.}\label{fig:4k_upb}
	\end{figure}
	
	Since the complete graph on $2k$ vertices has a $1$-factorization \cite[Theorem~9.1]{Har69}, on each of these remaining $p - 1 = 2k - 1$ parties we can connect each pair of vertices to exactly one other pair of vertices in such a way that the union of these $p$ orthogonality graphs is the complete graph, so the corresponding product states are mutually orthogonal. The fact that this product basis is also unextendible follows easily from its construction -- any product state can be orthogonal to at most $1$ state on the first party and at most $2$ of the states on each of the remaining $p-1$ parties, for a total of $1 + 2(p-1) = 2p-1$ states. Thus there is no product state orthogonal to all $s = 2p$ members of this product basis.
	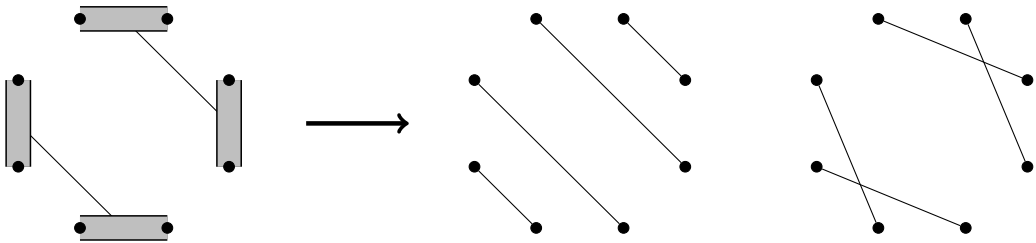
\begin{figure}[htb]
		\centering
		\begin{tikzpicture}[x=1.5cm, y=1.5cm, label distance=0cm]			
			\vertex[fill] (w0) at (0.617,-2.076) []{};
			\vertex[fill] (w1) at (0.076,-2.617) []{};
			\vertex[fill] (w2) at (0.076,-3.383) []{};
			\vertex[fill] (w3) at (0.617,-3.924) []{};
			\vertex[fill] (w4) at (1.383,-3.924) []{};
			\vertex[fill] (w5) at (1.924,-3.383) []{};
			\vertex[fill] (w6) at (1.924,-2.617) []{};
			\vertex[fill] (w7) at (1.383,-2.076) []{};

			\path 
				(w0) edge (w5)
				(w7) edge (w6)
				(w1) edge (w4)
				(w2) edge (w3)
			;
				
			\draw [->,line width=0.06cm] (-1.4,-3) -- (-0.5,-3);
				
		    \draw[draw=black] (-3,-2.076) -- (-2.076,-3);
		    \draw[draw=black] (-3.924,-3) -- (-3,-3.924);

			\filldraw[fill=lightgray,line width=0.34cm,line join=round,draw=black] (-3.924,-2.617) -- (-3.924,-3.383) -- cycle;
		    \filldraw[fill=lightgray,line width=0.30cm,line join=round,draw=lightgray] (-3.924,-2.617) -- (-3.924,-3.383) -- cycle;
			\filldraw[fill=lightgray,line width=0.34cm,line join=round,draw=black] (-3.383,-3.924) -- (-2.617,-3.924) -- cycle;
		    \filldraw[fill=lightgray,line width=0.30cm,line join=round,draw=lightgray] (-3.383,-3.924) -- (-2.617,-3.924) -- cycle;
			\filldraw[fill=lightgray,line width=0.34cm,line join=round,draw=black] (-2.076,-3.383) -- (-2.076,-2.617) -- cycle;
		    \filldraw[fill=lightgray,line width=0.30cm,line join=round,draw=lightgray] (-2.076,-3.383) -- (-2.076,-2.617) -- cycle;
			\filldraw[fill=lightgray,line width=0.34cm,line join=round,draw=black] (-3.383,-2.076) -- (-2.617,-2.076) -- cycle;
			\filldraw[fill=lightgray,line width=0.30cm,line join=round,draw=lightgray] (-3.383,-2.076) -- (-2.617,-2.076) -- cycle;
			
			\vertex[fill] (x0) at (-3.383,-2.076) []{};
			\vertex[fill] (x1) at (-3.924,-2.617) []{};
			\vertex[fill] (x2) at (-3.924,-3.383) []{};
			\vertex[fill] (x3) at (-3.383,-3.924) []{};
			\vertex[fill] (x4) at (-2.617,-3.924) []{};
			\vertex[fill] (x5) at (-2.076,-3.383) []{};
			\vertex[fill] (x6) at (-2.076,-2.617) []{};
			\vertex[fill] (x7) at (-2.617,-2.076) []{};
				
			\vertex[fill] (y0) at (4.383,-2.076) []{};
			\vertex[fill] (y1) at (4.924,-2.617) []{};
			\vertex[fill] (y2) at (4.924,-3.383) []{};
			\vertex[fill] (y3) at (4.383,-3.924) []{};
			\vertex[fill] (y4) at (3.617,-3.924) []{};
			\vertex[fill] (y5) at (3.076,-3.383) []{};
			\vertex[fill] (y6) at (3.076,-2.617) []{};
			\vertex[fill] (y7) at (3.617,-2.076) []{};

			\path 
				(y0) edge (y2)
				(y7) edge (y1)
				(y3) edge (y5)
				(y4) edge (y6)
			;
		\end{tikzpicture}
		\caption{An example of how to ``split'' one of the parties in the orthogonality graph of Figure~\ref{fig:4k_upb} into two while preserving unextendibility.}\label{fig:split_og}
	\end{figure}

	To generalize this construction to the $p+1 \leq s < 2p$ case, we modify some of the last $p-1$ parties in the $s = 2p$ construction above by ``splitting'' one party into two. To ``split'' a qubit, replace each pair of orthogonal states of the form $\{\ket{a},\ket{a},\ket{\overline{a}},\ket{\overline{a}}\}$ with the two-qubit states $\{\ket{aa},\ket{bb},\ket{\overline{ab}},\ket{\overline{ba}}\}$ (see Figure~\ref{fig:split_og}). This procedure is easily-verified to preserve unextendibility and orthogonality, so the resulting set of product states is a UPB.

	Furthermore, since this procedure keeps $s$ the same but increases the number of parties by $1$, and we can split anywhere from $1$ up to $p-1 = s/2 - 1$ orthogonality graphs in this way, we can construct an unextendible product basis of $s$ states for any number of parties $p$ from $s/2$ up to $s-1$, as desired (see Figure~\ref{fig:split_og_ex}).
	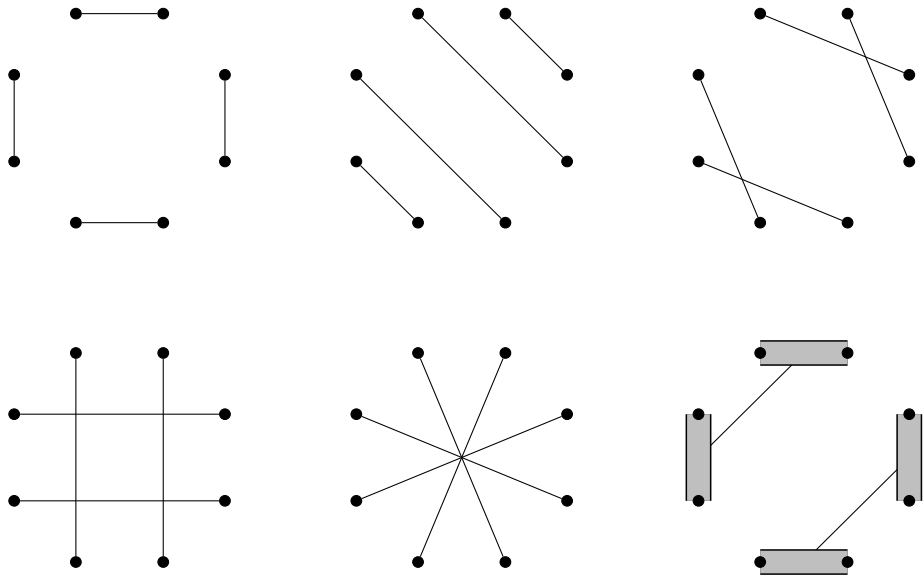
\begin{figure}[htb]
		\centering
		\begin{tikzpicture}[x=1.5cm, y=1.5cm, label distance=0cm]
			\vertex[fill] (v0) at (-3.383,0.824) []{};
			\vertex[fill] (v1) at (-3.924,0.283) []{};
			\vertex[fill] (v2) at (-3.924,-0.483) []{};
			\vertex[fill] (v3) at (-3.383,-1.024) []{};
			\vertex[fill] (v4) at (-2.617,-1.024) []{};
			\vertex[fill] (v5) at (-2.076,-0.483) []{};
			\vertex[fill] (v6) at (-2.076,0.283) []{};
			\vertex[fill] (v7) at (-2.617,0.824) []{};			

			\path 
				(v0) edge (v3)
				(v7) edge (v4)
				(v1) edge (v6)
				(v2) edge (v5)
			;

			\vertex[fill] (w0) at (-0.383,0.824) []{};
			\vertex[fill] (w1) at (-0.924,0.283) []{};
			\vertex[fill] (w2) at (-0.924,-0.483) []{};
			\vertex[fill] (w3) at (-0.383,-1.024) []{};
			\vertex[fill] (w4) at (0.383,-1.024) []{};
			\vertex[fill] (w5) at (0.924,-0.483) []{};
			\vertex[fill] (w6) at (0.924,0.283) []{};
			\vertex[fill] (w7) at (0.383,0.824) []{};

			\path 
				(w0) edge (w4)
				(w1) edge (w5)
				(w2) edge (w6)
				(w7) edge (w3)
			;

		    \draw[draw=black] (3,0.824) -- (2.076,-0.1);
		    \draw[draw=black] (3.924,-0.1) -- (3,-1.024);

			\filldraw[fill=lightgray,line width=0.34cm,line join=round,draw=black] (3.924,0.283) -- (3.924,-0.483) -- cycle;
		    \filldraw[fill=lightgray,line width=0.30cm,line join=round,draw=lightgray] (3.924,0.283) -- (3.924,-0.483) -- cycle;
			\filldraw[fill=lightgray,line width=0.34cm,line join=round,draw=black] (3.383,-1.024) -- (2.617,-1.024) -- cycle;
		    \filldraw[fill=lightgray,line width=0.30cm,line join=round,draw=lightgray] (3.383,-1.024) -- (2.617,-1.024) -- cycle;
			\filldraw[fill=lightgray,line width=0.34cm,line join=round,draw=black] (2.076,-0.483) -- (2.076,0.283) -- cycle;
		    \filldraw[fill=lightgray,line width=0.30cm,line join=round,draw=lightgray] (2.076,-0.483) -- (2.076,0.283) -- cycle;
			\filldraw[fill=lightgray,line width=0.34cm,line join=round,draw=black] (3.383,0.824) -- (2.617,0.824) -- cycle;
			\filldraw[fill=lightgray,line width=0.30cm,line join=round,draw=lightgray] (3.383,0.824) -- (2.617,0.824) -- cycle;
			
			\vertex[fill] (y0) at (3.383,0.824) []{};
			\vertex[fill] (y1) at (3.924,0.283) []{};
			\vertex[fill] (y2) at (3.924,-0.483) []{};
			\vertex[fill] (y3) at (3.383,-1.024) []{};
			\vertex[fill] (y4) at (2.617,-1.024) []{};
			\vertex[fill] (y5) at (2.076,-0.483) []{};
			\vertex[fill] (y6) at (2.076,0.283) []{};
			\vertex[fill] (y7) at (2.617,0.824) []{};

			\vertex[fill] (v20) at (-3.383,3.824) []{};
			\vertex[fill] (v21) at (-3.924,3.283) []{};
			\vertex[fill] (v22) at (-3.924,2.517) []{};
			\vertex[fill] (v23) at (-3.383,1.976) []{};
			\vertex[fill] (v24) at (-2.617,1.976) []{};
			\vertex[fill] (v25) at (-2.076,2.517) []{};
			\vertex[fill] (v26) at (-2.076,3.283) []{};
			\vertex[fill] (v27) at (-2.617,3.824) []{};			

			\path 
				(v20) edge (v27)
				(v22) edge (v21)
				(v24) edge (v23)
				(v26) edge (v25)
			;

			\vertex[fill] (w20) at (-0.383,3.824) []{};
			\vertex[fill] (w21) at (-0.924,3.283) []{};
			\vertex[fill] (w22) at (-0.924,2.517) []{};
			\vertex[fill] (w23) at (-0.383,1.976) []{};
			\vertex[fill] (w24) at (0.383,1.976) []{};
			\vertex[fill] (w25) at (0.924,2.517) []{};
			\vertex[fill] (w26) at (0.924,3.283) []{};
			\vertex[fill] (w27) at (0.383,3.824) []{};

			\path 
				(w20) edge (w25)
				(w27) edge (w26)
				(w21) edge (w24)
				(w22) edge (w23)
			;

			\vertex[fill] (y20) at (3.383,3.824) []{};
			\vertex[fill] (y21) at (3.924,3.283) []{};
			\vertex[fill] (y22) at (3.924,2.517) []{};
			\vertex[fill] (y23) at (3.383,1.976) []{};
			\vertex[fill] (y24) at (2.617,1.976) []{};
			\vertex[fill] (y25) at (2.076,2.517) []{};
			\vertex[fill] (y26) at (2.076,3.283) []{};
			\vertex[fill] (y27) at (2.617,3.824) []{};

			\path 
				(y20) edge (y22)
				(y27) edge (y21)
				(y23) edge (y25)
				(y24) edge (y26)
			;
		\end{tikzpicture}
		\caption{An orthogonality graph of a UPB in the $p = 6, s = 8$ case of Theorem~\ref{thm:qubit_4k_upbs}, constructed by ``splitting'' the bottom-left and bottom-center qubits of the orthogonality graph in Figure~\ref{fig:4k_upb}. Similarly, splitting only $1$ qubit results in the UPB in the $p = 5, s = 8$ case presented in Section~\ref{section:5qubit}, while splitting $3$ parties would result in a UPB in the $p = 7, s = 8$ case.}\label{fig:split_og_ex}
	\end{figure}
\end{proof}

Note that Theorem~\ref{thm:qubit_4k_upbs} says nothing about the existence or nonexistence of qubit UPBs in the case when $p \equiv 1 \, (\text{mod } 4)$ and $s = 2p + 2$. The smallest such case is when $p = 5$ and $s = 12$, and in fact a UPB \emph{does} exist in this case, simply by combining two copies of the $4$-qubit UPB of size $6$. This leaves the $p = 9, s = 20$ case as the smallest case where the existence of a qubit UPB whose size is a multiple of $4$ is unknown.

\section{Conclusions and Outlook}\label{section:conclusions}

We have investigated the structure of qubit unextendible product bases by completely characterizing them in the $4$-qubit case and deriving many new results concerning the (non)existence of qubit UPBs of given sizes. This work has many immediate applications and consequences.

For example, there is a long history of trying to determine the possible ranks of bound entangled states \cite{HSTT03,Cla06,Ha07,KO12,CD13}, however very little is known about this question in the multipartite case. By using the standard method of creating a bound entangled state from a UPB (i.e., if $\{\ket{v_1},\ldots,\ket{v_s}\}$ is a UPB then $\rho := I - \sum_{i=1}^s \ketbra{v_i}{v_i}$ is a multiple of a bound entangled state), we can use our results to construct $p$-qubit bound entangled states of many different ranks. For example, Theorem~\ref{thm:many_qubit_sizes} immediately implies that (for $p \geq 7$) there exist $p$-qubit bound entangled states of every rank from $6$ through $2^p - \frac{p^2 + 3p - 30}{2}$, inclusive.

It is also known that some qubit UPBs can be used to construct locally indistinguishable subspaces of the same size \cite{DXY10}. It is unknown whether or not \emph{all} qubit UPBs span a locally indistinguishable subspace, so it might be the case that the hundreds of new UPBs found in this work contain a counter-example. Alternatively, if it turns out that qubit UPBs do always span a locally indistinguishable subspace, it would follow that there exist such subspaces of all sizes found in this work.

However, some notable open problems about qubit UPBs remain:
\begin{enumerate}
	\item Does there exist a $p$-qubit UPB of size $2^p - 5$ when $p \geq 5$? Our computer search showed that the answer is ``no'' when $p = 4$, but we still do not know of a simple reason for why this is the case. It is worth noting that if the answer is ``yes'' for any particular value of $p$ then it must be ``yes'' for all larger values of $p$ as well, by Proposition~\ref{prop:qubit_add_together}. Closely related to this question is whether or not there exist $p$-qubit bound entangled states of rank $5$.
	
	\item Does there exist a $p$-qubit UPB of size $2p+2$ when $p \equiv 1 \, (\text{mod } 4)$? That is, can we fill in the hole in Theorem~\ref{thm:qubit_4k_upbs}?
	
	\item What is the true maximum ``gap size'' $g_p$ as described in Section~\ref{section:manyqubit}? The best bounds that we have so far are $g_p \geq 1$ when $p \geq 5$ is odd, $g_p \leq 3$ when $p \not\equiv 1 \, (\text{mod } 4)$, and $g_p \leq 7$ always.
	
	\item There are also many other cases where the existence of an $s$-state $p$-qubit UPB is unknown, since very little is known about the existence of such UPBs when $f(p) < s < (p^2 + 3p - 30)/2$ and $s$ is not a multiple of $4$ (see Table~\ref{tab:qubit_summary}).
\end{enumerate}
\begin{table}\vspace*{-0.4in}
\begin{center}
	\def\arraystretch{1.05}
    \begin{tabular}{@{}c c c c c c c c}
    \toprule
     & \multicolumn{7}{c}{number of qubits} \\  \cmidrule{2-8}
    size & $1$ & $2$ & $3$ & $4$ & $5$ & $6$ & $7$ \\ \midrule
	1 & \cellcolor{redback} & \cellcolor{redback} & \cellcolor{redback} & \cellcolor{redback} & \cellcolor{redback} & \cellcolor{redback} & \cellcolor{redback} \\
	2 & \cellcolor{greenback}\cmark & \cellcolor{redback} & \cellcolor{redback} & \cellcolor{redback} & \cellcolor{redback} & \cellcolor{redback} & \cellcolor{redback} \\
	3 & \cellcolor{redback} & \cellcolor{redback} & \cellcolor{redback} & \cellcolor{redback} & \cellcolor{redback} & \cellcolor{redback} & \cellcolor{redback} \\
	4 & \cellcolor{redback} & \cellcolor{greenback}\cmark & \cellcolor{greenback}\cmark & \cellcolor{redback} & \cellcolor{redback} & \cellcolor{redback} & \cellcolor{redback} \\
	5 & \cellcolor{redback} & \cellcolor{redback} & \cellcolor{redback} & \cellcolor{redback} & \cellcolor{redback} & \cellcolor{redback} & \cellcolor{redback} \\
	6 & \cellcolor{redback} & \cellcolor{redback} & \cellcolor{redback} & \cellcolor{greenback}\cmark & \cellcolor{greenback}\cmark & \cellcolor{redback} & \cellcolor{redback} \\
	7 & \cellcolor{redback} & \cellcolor{redback} & \cellcolor{redback} & \cellcolor{greenback}\cmark & \cellcolor{redback} & \cellcolor{redback} & \cellcolor{redback} \\
	8 & \cellcolor{redback} & \cellcolor{redback} & \cellcolor{greenback}\cmark & \cellcolor{greenback}\cmark & \cellcolor{greenback}\cmark & \cellcolor{greenback}\cmark & \cellcolor{greenback}\cmark \\
	9 & \cellcolor{redback} & \cellcolor{redback} & \cellcolor{redback} & \cellcolor{greenback}\cmark & \cellcolor{greenback}\cmark & \cellcolor{greenback}\cmark & \cellcolor{redback} \\
    10 & \cellcolor{redback} & \cellcolor{redback} & \cellcolor{redback} & \cellcolor{greenback}\cmark & \cellcolor{greenback}\cmark & \cellcolor{yellowback}? & \cellcolor{yellowback}? \\
    11 & \cellcolor{redback} & \cellcolor{redback} & \cellcolor{redback} & \cellcolor{redback} & \cellcolor{yellowback}? & \cellcolor{yellowback}? & \cellcolor{yellowback}? \\
    12 & \cellcolor{redback} & \cellcolor{redback} & \cellcolor{redback} & \cellcolor{greenback}\cmark & \cellcolor{greenback}\cmark & \cellcolor{greenback}\cmark & \cellcolor{greenback}\cmark \\
    13 & \cellcolor{redback} & \cellcolor{redback} & \cellcolor{redback} & \cellcolor{redback} & \cellcolor{greenback}\cmark & \cellcolor{yellowback}? & \cellcolor{yellowback}? \\
    14 & \cellcolor{redback} & \cellcolor{redback} & \cellcolor{redback} & \cellcolor{redback} & \cellcolor{greenback}\cmark & \cellcolor{greenback}\cmark & \cellcolor{yellowback}? \\
    15 & \cellcolor{redback} & \cellcolor{redback} & \cellcolor{redback} & \cellcolor{redback} & \cellcolor{greenback}\cmark & \cellcolor{greenback}\cmark & \cellcolor{yellowback}? \\
    16 & \cellcolor{redback} & \cellcolor{redback} & \cellcolor{redback} & \cellcolor{greenback}\cmark & \cellcolor{greenback}\cmark & \cellcolor{greenback}\cmark & \cellcolor{greenback}\cmark \\
    17--18 & \cellcolor{redback} & \cellcolor{redback} & \cellcolor{redback} & \cellcolor{redback} & \cellcolor{greenback}\cmark & \cellcolor{greenback}\cmark & \cellcolor{greenback}\cmark \\
    19 & \cellcolor{redback} & \cellcolor{redback} & \cellcolor{redback} & \cellcolor{redback} & \cellcolor{greenback}\cmark & \cellcolor{greenback}\cmark & \cellcolor{yellowback}? \\
    20--26 & \cellcolor{redback} & \cellcolor{redback} & \cellcolor{redback} & \cellcolor{redback} & \cellcolor{greenback}\cmark & \cellcolor{greenback}\cmark & \cellcolor{greenback}\cmark \\
    27 & \cellcolor{redback} & \cellcolor{redback} & \cellcolor{redback} & \cellcolor{redback} & \cellcolor{yellowback}? & \cellcolor{greenback}\cmark & \cellcolor{greenback}\cmark \\
    28 & \cellcolor{redback} & \cellcolor{redback} & \cellcolor{redback} & \cellcolor{redback} & \cellcolor{greenback}\cmark & \cellcolor{greenback}\cmark & \cellcolor{greenback}\cmark \\
    29--31 & \cellcolor{redback} & \cellcolor{redback} & \cellcolor{redback} & \cellcolor{redback} & \cellcolor{redback} & \cellcolor{greenback}\cmark & \cellcolor{greenback}\cmark \\
    32 & \cellcolor{redback} & \cellcolor{redback} & \cellcolor{redback} & \cellcolor{redback} & \cellcolor{greenback}\cmark & \cellcolor{greenback}\cmark & \cellcolor{greenback}\cmark \\
    33--58 & \cellcolor{redback} & \cellcolor{redback} & \cellcolor{redback} & \cellcolor{redback} & \cellcolor{redback} & \cellcolor{greenback}\cmark & \cellcolor{greenback}\cmark \\
    59 & \cellcolor{redback} & \cellcolor{redback} & \cellcolor{redback} & \cellcolor{redback} & \cellcolor{redback} & \cellcolor{yellowback}? & \cellcolor{greenback}\cmark \\
    60 & \cellcolor{redback} & \cellcolor{redback} & \cellcolor{redback} & \cellcolor{redback} & \cellcolor{redback} & \cellcolor{greenback}\cmark & \cellcolor{greenback}\cmark \\
    61--63 & \cellcolor{redback} & \cellcolor{redback} & \cellcolor{redback} & \cellcolor{redback} & \cellcolor{redback} & \cellcolor{redback} & \cellcolor{greenback}\cmark \\
    64 & \cellcolor{redback} & \cellcolor{redback} & \cellcolor{redback} & \cellcolor{redback} & \cellcolor{redback} & \cellcolor{greenback}\cmark & \cellcolor{greenback}\cmark \\
    65--122 & \cellcolor{redback} & \cellcolor{redback} & \cellcolor{redback} & \cellcolor{redback} & \cellcolor{redback} & \cellcolor{redback} & \cellcolor{greenback}\cmark \\
    123 & \cellcolor{redback} & \cellcolor{redback} & \cellcolor{redback} & \cellcolor{redback} & \cellcolor{redback} & \cellcolor{redback} & \cellcolor{yellowback}? \\
    124 & \cellcolor{redback} & \cellcolor{redback} & \cellcolor{redback} & \cellcolor{redback} & \cellcolor{redback} & \cellcolor{redback} & \cellcolor{greenback}\cmark \\
    125--127 & \cellcolor{redback} & \cellcolor{redback} & \cellcolor{redback} & \cellcolor{redback} & \cellcolor{redback} & \cellcolor{redback} & \cellcolor{redback} \\
    128 & \cellcolor{redback} & \cellcolor{redback} & \cellcolor{redback} & \cellcolor{redback} & \cellcolor{redback} & \cellcolor{redback} & \cellcolor{greenback}\cmark \\\bottomrule
    \end{tabular}
    \caption{(color online) A summary of what sizes of UPBs are possible on small numbers of qubits. Empty \colorbox{redback}{red} cells indicate that no UPB of that size is possible on that number of qubits, while \colorbox{greenback}{green} checkmarks (\cmark) indicate that an explicit UPB of that size is known. \colorbox{yellowback}{Yellow} question marks (?) indicate that it is currently unknown whether or not such a UPB exists.}\label{tab:qubit_summary}
\end{center}
\end{table}

\vspace{0.1in} \noindent{\bf Acknowledgements.} The author thanks Remigiusz Augusiak and Jianxin Chen for helpful conversations related to this work. The author was supported by the Natural Sciences and Engineering Research Council of Canada. 

\bibliographystyle{alpha}
\bibliography{quantum}

\section*{Appendix A: Details of the Computation}\label{section:computation}

Since it is not feasible to find all $4$-qubit UPBs via naive brute force search, our search is split into three steps so that most of the work can be done in parallel. The code used to carry out the search can be downloaded from \cite{JohUPBCode14}. The code that does the time-consuming parts of the computation is written in C, and then some post-processing is done in Maple.

\subsection*{Step 1: Finding Potential Bipartite Graph Decompositions}\label{section:computation_step1}

Recall from Section~\ref{section:orthog_graphs} that the orthogonality graph of each qubit of a UPB can decomposed as the disjoint union of complete bipartite graphs. The first step is to try to find sizes of complete bipartite graphs that could possibly correspond to UPBs, without considering the actual placement of those complete bipartite graphs within each qubit.

That is, we first search for positive integers $n_1,\ldots,n_p$ and $\{a^{(i)}_1,b^{(i)}_1,\ldots,a^{(i)}_{n_i},b^{(i)}_{n_i}\}_{i=1}^{p}$ such that there could possibly be a UPB with the property that the orthogonality graph of its $i$-th qubit ($1 \leq i \leq p$) is the disjoint union of $K_{a^{(i)}_1,b^{(i)}_1},\ldots,K_{a^{(i)}_{n_i},b^{(i)}_{n_i}}$, where $n_i$ is the total number of complete bipartite graphs present in the orthogonality graph of the $i$-th qubit.

There are many simple necessary conditions that the $a^{(i)}_j$'s and $b^{(i)}_j$'s must satisfy in order for there to be a corresponding UPB. For example, we clearly must have $\sum_{j=1}^{n_i}(a^{(i)}_j+b^{(i)}_j) = s$ for all $i$, since there must be exactly $s$ vertices in each qubit's orthogonality graph. Also, it is straightforward to see that $a^{(i)}_j,b^{(i)}_j \leq s - p$ for all $i,j$, since otherwise we could find a product state that is orthogonal to $s-p+1$ states of the UPB on one qubit and $1$ state of the UPB on each of the remaining qubits, for a total of $(s-p+1) + (p-1) = s$ states, which implies extendibility. Furthermore, since $K_{a^{(i)}_j,b^{(i)}_j}$ has $a^{(i)}_j b^{(i)}_j$ edges and we need at least $s(s-1)/2$ edges in order for the corresponding states to be mutually orthogonal, we must have $\sum_{i,j}a^{(i)}_j b^{(i)}_j \geq s(s-1)/2$.

There are also some less obvious restrictions that we can place on our search space, as described by the following lemmas.
\begin{lemma}\label{lem:reverse_combine}
	Suppose that the constants $\{a^{(i)}_1,b^{(i)}_1,\ldots,a^{(i)}_{n_i},b^{(i)}_{n_i}\}_{i=1}^{p}$ (defined above) correspond to a $p$-qubit UPB. There exists a particular $i$ such that $n_i = 1$ if and only if there exist $(p-1)$-qubit UPBs of size $a^{(i)}_{1}$ and $b^{(i)}_{1}$.
\end{lemma}
\begin{proof}
	The ``if'' direction of this lemma is actually just a rewording of Proposition~\ref{prop:qubit_add_together}. If there are $(p-1)$-qubit UPBs of size $a^{(i)}_{1}$ and $b^{(i)}_{1}$ then the method of construction given in the proof of Proposition~\ref{prop:qubit_add_together} gives a $p$-qubit UPB with $n_p = 1$ (i.e., only one basis of $\mathbb{C}^2$ is used on the $p$-th qubit).
	
	For the ``only if'' direction of the proof, suppose for a contradiction that $n_i = 1$ for some $i$, which means that there exists a basis $\{\ket{a},\ket{\overline{a}}\}$ of $\mathbb{C}^2$ such that $a^{(i)}_{1}$ of the states in the $p$-qubit UPB are equal to $\ket{a}$ on the $i$-th qubit, and the remaining $b^{(i)}_{1}$ states are equal to $\ket{\overline{a}}$ on the $i$-th qubit. It is straightforward to check that those sets of $a^{(i)}_{1}$ and  $b^{(i)}_{1}$ states form $(p-1)$-qubit UPBs if we remove their $i$-th qubits.
\end{proof}

\begin{lemma}\label{lem:search_reduce}
	Suppose that the constants $\{a^{(i)}_1,b^{(i)}_1,\ldots,a^{(i)}_{n_i},b^{(i)}_{n_i}\}_{i=1}^{p}$ (defined above) correspond to a UPB. Fix any permutation $\sigma : \{1,\ldots,p\} \rightarrow \{1,\ldots,p\}$, let $t_1$ be any of the $a^{(\sigma(1))}_j$'s or $b^{(\sigma(1))}_j$'s, and define the constants $t_2,\ldots,t_p$ recursively as follows:
	\begin{align*}
		t_k & := \min_{c_j,d_j} \Big\{ \max_{j} \Big\{ a^{(\sigma(k))}_j - c_j, b^{(\sigma(k))}_j - d_j : c_j,d_j \geq 0 \text{ are integers}, \sum_j c_j + \sum_j d_j = t_{k-1}\Big\} \Big\}.
	\end{align*}
	Then $\sum_{i=1}^p t_i \leq s - 1$.
\end{lemma}

Before proving Lemma~\ref{lem:search_reduce}, we note that it is actually slightly more intuitive than it appears at first glance, as it is just a generalization of the fact that $a^{(i)}_j,b^{(i)}_j \leq s - p$ for all $i,j$ that takes into account how large the $a^{(i)}_j$'s and $b^{(i)}_j$'s are on more than one party. For example, the lemma says that there can not be an $8$-state UPB on $5$ qubits such that the orthogonality graphs of its first two qubits each decompose as $K_{3,3} \cup K_{1,1}$, since we could then choose $t_1 = 3$, which gives $t_2 = 2$ and $t_3 = t_4 = t_5 = 1$ and we have $\sum_{i=1}^p t_i = 8 > s-1 = 7$.

We can come to the same conclusion in a more intuitive manner by noting that we can always find a product state orthogonal to $3$ states on the first qubit and at least $2$ more states on the second qubit (and of course $1$ state on each of the remaining qubits), since the groups of $3$ equal states on the first two qubits can not overlap ``too much''.
\begin{proof}[Proof of Lemma~\ref{lem:search_reduce}]
	The result follows from simply observing that we can find a product state that is orthogonal to $t_1$ members of the UPB on party $\sigma(1)$, $t_2$ more members of the UPB on qubit $\sigma(2)$, and so on. Thus unextendibility implies that $\sum_{i=1}^p t_i \leq s - 1$.
\end{proof}

By making use of Lemmas~\ref{lem:reverse_combine} and~\ref{lem:search_reduce}, as well as the other basic restrictions mentioned earlier, we are able to perform a brute-force search that gives a list of potential values of the $a^{(i)}_j$'s and $b^{(i)}_j$'s. However, many of these results do not actually lead to UPBs. For example, in the case of $11$-state $4$-qubit UPBs, the above restrictions give a list of $14449$ possible values for the $a^{(i)}_j$'s and $b^{(i)}_j$'s, such as the following:
\begin{align*}
	\text{Qubit 1:} & \ K_{4,3} \cup K_{1,1} \cup K_{1,1} \\
	\text{Qubit 2:} & \ K_{4,1} \cup K_{3,3} \\
	\text{Qubit 3:} & \ K_{3,3} \cup K_{3,2} \\
	\text{Qubit 4:} & \ K_{3,3} \cup K_{3,2}.
\end{align*}
However, it turns out that there are not actually any UPBs whose orthogonality graph has such a decomposition (nor any of the other $14448$ potential decompositions), which is determined in the next step of the computation.

\subsection*{Step 2: Checking Each Decomposition}\label{section:computation_step2}

The second step in the computation is much more time-consuming---it consists of checking each of the decompositions found in the first step to see if there are actually any UPBs with such a decomposition. However, searching each of these different decompositions can be done in parallel, which greatly speeds up the process. This part of the search is done via standard brute force and is fairly straightforward.

\subsection*{Step 3: Sorting the Results}\label{section:computation_step3}

The third (and final) step in the computation is to sort the UPBs into equivalence classes based on their orthogonality graphs. This step is necessary because many of the UPBs found in step 2 are actually equivalent to each other (i.e., they are the same up to relabeling local bases and permuting states and qubits). This step is quick enough that it is also done by fairly standard brute force: for each UPB found in step 2, all possible relabelings and permutations of the UPB are generated and checked against all other UPBs found in step 2. If a match is found, then one of the two UPBs is discarded, since they are equivalent. Note that no particular preference is given to \emph{which} UPB is discarded.

\section*{Appendix B: Comparison of the Lower Bounds of Theorem~\ref{thm:many_qubit_sizes}}\label{section:bound_comparison}

Recall from the discussion surrounding Theorem~\ref{thm:many_qubit_sizes} that we claimed that $\sum_{k=4}^{p-1}f(k)$ is a lower bound of $(p^2+3p-30)/2$, and furthermore than these two quantities never differ by more than $2$. We now prove these claims explicitly.
\begin{prop}\label{prop:fk_sum}
	Let $p \geq 7$ be an integer. Then
	\begin{align*}
		\sum_{k=4}^{p-1}f(k) \leq \frac{p^2+3p-30}{2} \leq \sum_{k=4}^{p-1}f(k) + 2.
	\end{align*}
\end{prop}
\begin{proof}
	We prove the result by induction. We first prove six base cases by noting that, for $p = 7, 8, 9, 10, 11, 12$ we have $\sum_{k=4}^{p-1}f(k) = 20, 28, 39, 49, 61, 73$ and $(p^2+3p-30)/2 = 20, 29, 39, 50, 62, 75$, so the result holds in these cases.
	
	For the inductive step, our goal is to show that, if the result holds for a fixed value of $p$ (with $p \geq 9$), then it holds for $p+4$ as well. To this end, note that it follows from the definition given in Equation~\eqref{eq:qubit_min_size} that $\sum_{k=p}^{p+3}f(k) = 4p + 14$ for all $p \geq 9$. Thus, by making use of the inductive hypothesis, we have
	\begin{align*}
		& \ \ \sum_{k=4}^{p-1}f(k) \leq \frac{p^2+3p-30}{2} \leq \sum_{k=4}^{p-1}f(k) + 2 \\
		\Longrightarrow & \ \ \sum_{k=4}^{p-1}f(k) + (4p+14) \leq \frac{p^2+3p-30}{2} + (4p+14) \leq \sum_{k=4}^{p-1}f(k) + 2 + (4p+14) \\
		\Longrightarrow & \ \ \sum_{k=4}^{p+3}f(k) \leq \frac{(p+4)^2+3(p+4)-30}{2} \leq \sum_{k=4}^{p+3}f(k) + 2,
	\end{align*}
	which completes the inductive step and the proof.
\end{proof}
\end{document}